\documentclass[11pt]{article}

\usepackage{amsmath,amsthm,amssymb,dsfont,enumerate,tikz,hyperref}
\usepackage{setspace}
\usepackage{comment}
\usepackage{tabu}
\usepackage[margin=1.25in]{geometry}
\usepackage{subcaption,graphicx}

\newtheorem{theorem}{Theorem}[section]
\newtheorem{proposition}[theorem]{Proposition}
\newtheorem{lemma}[theorem]{Lemma}

\theoremstyle{definition}

\newtheorem{definition}[theorem]{Definition}

\newtheorem{remark}[theorem]{Remark}

\DeclareMathOperator*{\argmin}{arg\,min}

\renewcommand{\subset}{\subseteq}
\renewcommand{\hat}{\widehat}
\renewcommand{\tilde}{\widetilde}
\renewcommand{\epsilon}{\varepsilon}
\renewcommand{\Re}{{\rm Re}}
\renewcommand{\Im}{{\rm Im}}
\newcommand{\commentout}[1]{}
\def\sign{\text{sign}}
\def\supp{\text{supp}}
\def\Lip{{\rm Lip}}
\def\TV{{\rm TV}}

\def\<{\langle}
\def\>{\rangle}
\def\({\Big(}
\def\){\Big)}
\def\diag{\text{diag}}
\def\calA{\mathcal{A}}
\def\C{\mathbb{C}}

\def\calE{\mathcal{E}}
\def\calF{\mathcal{F}}

\def\calZ{\mathcal{Z}}
\def\calI{\mathcal{I}}
\def\calM{\mathcal{M}}
\def\calIM{\mathcal{I}^M}

\def\calN{\mathcal{N}}

\def\R{\mathbb{R}}

\def\T{\mathbb{T}}

\def\vol{\text{vol}}
\def\Z{\mathbb{Z}}

\numberwithin{equation}{section}

\begin{document}

\title{Quantization for spectral super-resolution}

\author{C.~Sinan G\"unt\"urk\thanks{Courant Institute of Mathematical Sciences, New York University. Email: gunturk@cims.nyu.edu} \and Weilin Li\thanks{Courant Institute of Mathematical Sciences, New York University. Email: weilinli@cims.nyu.edu}}

\maketitle

\begin{abstract}
We show that the method of distributed noise-shaping beta-quantization offers superior performance for the problem of spectral super-resolution with quantization whenever there is redundancy in the number of measurements. More precisely, we define the oversampling ratio $\lambda$ as the largest integer such that $\lfloor M/\lambda\rfloor - 1\geq 4/\Delta$, where $M$ denotes the number of Fourier measurements and $\Delta$ is the minimum separation distance associated with the atomic measure to be resolved. We prove that for any number $K\geq 2$ of quantization levels available for the real and imaginary parts of the measurements, our quantization method combined with either TV-min/BLASSO or ESPRIT guarantees reconstruction accuracy of order $O(M^{1/4}\lambda^{5/4} K^{- \lambda/2})$ and $O(M^{3/2} \lambda^{1/2} K^{- \lambda})$ respectively, where the implicit constants  are independent of $M$, $K$ and $\lambda$. In contrast, naive rounding or memoryless scalar quantization for the same alphabet offers a guarantee of order $O(M^{-1}K^{-1})$ only, regardless of the reconstruction algorithm.
\end{abstract}

\noindent {\bf Keywords:} Quantization, super-resolution, spectral estimation, total variation, ESPRIT

\medskip 

\noindent {\bf MSC2020:} 94A12, 94A20

\section{Introduction}

Analog-to-digital conversion is inherently lossy. It typically consists of two stages, {\it sampling} and {\it quantization}. The sampling stage produces a stream of scalar samples and the quantization stage replaces each sample with an element of a discrete set, called the {\it (quantization) alphabet}. Ideally the sampling stage is lossless (or negligible) so that the distortion is only (or primarily) caused by quantization. In many applications, one is tasked with designing an analog-to-digital conversion scheme that minimizes the rate-distortion or reconstruction error. 

The naive approach to quantization is to round each scalar measurement to the nearest available level in the quantization alphabet, which is called  {\it memoryless scalar quantization} (MSQ). While MSQ is simple to implement in hardware and is suitable to be used with robust recovery algorithms in the high-resolution regime, its rate-distortion performance is suboptimal when the sampling map is redundant, i.e. when more measurements are collected compared to the minimal number needed for perfect reconstruction. The reason for this suboptimality is simple: the dimensionality of the manifold on which the measurements lie does not increase with the number of measurements once it exceeds this critical value, and therefore, the process of rounding these measurements to the nearest lattice points, i.e. MSQ, becomes wasteful as most lattice points are never utilized.

More efficient quantization methods achieve improved rate-distortion performance by utilizing some (or all) of these lattice points that are missed by MSQ. {\it Noise-shaping quantizers} such as {\it $\Sigma\Delta$ modulation} and {\it beta-quantization} fall into this category. The performance of $\Sigma\Delta$ has been extensively analyzed in the context of quantization for band-limited functions (e.g. \cite{daubechies2003approximating,gunturk2003one,gunturk2004approximating,deift2011optimal}), finite frame coefficients (e.g. \cite{benedetto2006sigma,benedetto2006second,wang2018sigma}), and compressive sampling (e.g. \cite{gunturk2013sobolev,saab2018quantization}). Beta-quantization is more recent and was applied to random Gaussian frames (e.g. \cite{chou2013beta,chou2016distributed}), harmonic Fourier frames (e.g. \cite{chou2017distributed,chou2015noise}), and fast binary embeddings (e.g. \cite{huynh2020fast}). This body of papers indicates that for many signal and data processing tasks, noise-shaping quantization is superior to that of MSQ. 

Extending this progression of work, we seek to use noise-shaping quantization for spectral super-resolution. Let $\calM(\T)$ be the collection of complex-valued atomic measures supported in $\T:=[0,1)$. Any $\mu\in \calM(\T)$ with at most $S$ atoms can be written in the form 
\begin{equation*}
	\mu :=\sum_{j=1}^S a_j\delta_{t_j},
	\quad a:=\big( a_j\big)_1^S \in\C^S,
	\quad T:=\{t_j\}_{j=1}^S\subset\T. 
\end{equation*}
For a fixed integer $M>0$, let $\calF_M$ be the operator which maps the measure $\mu$ to its first $M$ Fourier coefficients, where the $k$-th  coefficient is defined to be
\begin{equation*}
	\hat\mu(k)
	:=\int_0^1 e^{-2\pi ik t}\ d\mu(t)
	=\sum_{j=1}^S a_j e^{-2\pi i k t_j}. 
\end{equation*}
The {\it super-resolution problem} is to recover $\mu$ from noisy observations of $\calF_M\mu$. 

\begin{figure}
	\centering
	\begin{tikzpicture}[scale=1.3]
	\draw (-1,1) node {$\mu$}; 
	\draw[thick,->] (-.75,1) -- (-.25,1);
	\draw[thick] (0,0) -- (2,0) -- (2,2) -- (0,2) -- (0,0);
	\draw (1,1.5) node[anchor=north] {Sampling}; 
	\draw (1,1) node[anchor=north] {$\calF_M$};
	\draw[thick,->] (2.25,1) -- (3.25,1); 
	\draw[thick] (3.5,0) -- (5.5,0) -- (5.5,2) -- (3.5,2) -- (3.5,0);
	\draw (4.5,1.5) node[anchor=north] {Quantization}; 
	\draw (4.5,1) node[anchor=north] {$Q$}; 
	\draw[thick,->] (5.75,1) -- (6.75,1);
	\draw[thick] (7,0) -- (9,0) -- (9,2) -- (7,2) -- (7,0);
	\draw (8,1.5) node[anchor=north] {Reconstruction}; 
	\draw (8,1) node[anchor=north] {$D$}; 
	\draw[thick,->] (9.25,1) -- (9.75,1);
	\draw (10,1) node {$\tilde\mu$}; 
	\end{tikzpicture}
	\caption{For any unknown $\mu\in\calM$, the measurement system collects Fourier information $\calF_M\mu$, which are quantized to or encoded as $Q(\calF_M\mu)$. A decoder $D$ reconstructs an approximate measure $\tilde\mu$ from the quantized Fourier data.}
	\label{fig:prob}
\end{figure}
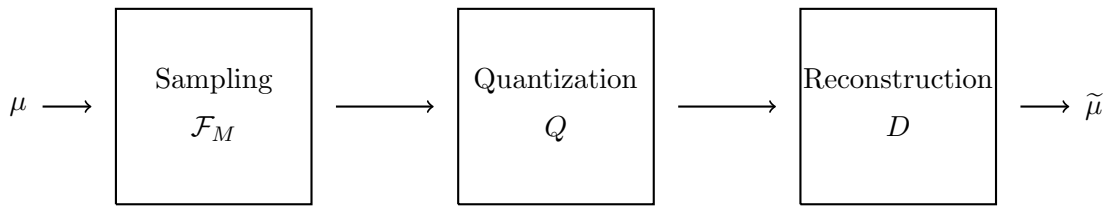

Let us informally describe the classes of measures we will consider. The support of $\mu$ plays an important role in the robustness of recovery. If it contains points that are too close to one another, recovery by some algorithm may still be possible, but small noise can lead to large reconstruction errors (e.g. \cite{donoho1992superresolution,demanet2015recoverability,li2021stable}). It is standard to define the minimum separation of an atomic $\nu$ by
\begin{equation}
\label{eq:minsep}
\Delta(\nu)
:=\min_{s\not=t,\, s,t\in\supp(\nu)} |s-t|_\T, 
\quad\text{where}\quad 
|t|_\T 
:=\min_{n \in \Z}~|t-n|, 
\end{equation}
and for a prescribed $\Delta>0$, we define a class of $\Delta$-separated measures 
\begin{equation*}
	\calM(\T,\Delta)
	:=\{\mu\in \calM(\T)\colon \Delta(\nu)\geq \Delta\}. 
\end{equation*}
All measures considered in this paper will belong to an appropriate subset of $\calM(\T,\Delta)$.

Each entry of the complex vector $y=\calF_M\mu$ must be quantized to an element of some complex alphabet $\calA$ before reconstruction can be done digitally. We consider Cartesian alphabets which are of the form $\calA_\R+i\calA_\R$ where $\calA_\R\subset\R$ contains $K\geq 2$ elements. Given a distortion function $\calE(\cdot,\cdot)$, and a class of atomic measures $\calM\subset \calM(\T,\Delta)$, the {\it (spectral) super-resolution quantization problem} is to select an alphabet $\calA\subset\C$ consisting of $K\times K$ elements, quantizer $Q\colon\C^M\to\calA^M$, and decoder $D\colon \calA^M\to \calM(\T)$ that minimizes the {\it rate-distortion}
\begin{equation*}
	\sup_{\mu\in\calM} \calE\big(\mu,\tilde\mu\big),
	\quad \text{where}\quad \tilde\mu:= D(Q(\calF_M \mu)). 
\end{equation*}
See Figure \ref{fig:prob} for an illustration of this setup. We will consider distortion functions $\calE(\cdot,\cdot)$ that include both amplitude and support errors, typically as a weighted sum. In this problem formulation, we have made the simplification that quantization is the only source of perturbation. 

The naive approach is to quantize the analog Fourier coefficients $y=\calF_M\mu$ via MSQ with a Cartesian alphabet of cardinality $K^2$ and feed them into an existing super-resolution algorithm. If $M\geq 4/\Delta(\mu)$, popular algorithms based on either convex optimization (e.g. TV-min/BLASSO) \cite{candes2013super,fernandez2013support,duval2015exact}  or subspace methods (e.g. ESPRIT) \cite{li2020super} can recover $\mu$ with error bounded by a universal constant times the $\ell^2$ norm of the quantization error, where the latter is bounded by $O(\sqrt{M}/ K)$. While the numerical evidence presented in this paper indicates that the $\sqrt M$ factor is artificial, it also shows that neither algorithm benefits from more samples beyond the $4/\Delta(\mu)$ threshold. It is important to mention that this inefficiency of MSQ exists in a wide array of problems besides super-resolution (e.g. \cite{daubechies2003approximating,gunturk2003one,gunturk2004approximating,gunturk2013sobolev}).

In this paper, we are primarily interested in the situation where a sampling/imaging device has already
collected more Fourier measurements than necessary. Our viewpoint
is to find what would be the best quantization algorithm given a set of measurements
and a fixed scalar quantizer. In this context, we advocate for and provide an alternative noise-shaping quantization, called beta-quantization, that is able to exploit the available redundancy (while using the same quantizer) and we establish rate-distortion bounds far superior to what is achievable with MSQ. More specifically, for fixed $\Delta>0$, assume there is an integer $\lambda\geq 1$ called the {\it over-sampling ratio} such that 
\begin{equation}
\label{eq:lambda}
\Big\lfloor \frac{M}{\lambda} \Big\rfloor \geq 1+\frac{4}{\Delta}.
\end{equation}
Hence, $\lambda$ roughly corresponds to ratio between the number of samples $M$ and $4/\Delta$. For an appropriate class of measures $\calM\subset\calM(\T,\Delta)$ and suitable distortion $\calE(\cdot,\cdot)$, we show that our beta-quantizer can be combined with modifications of existing super-resolution algorithms (TV-min/BLASSO, ESPRIT) to guarantee rate-distortion far surpassing that of MSQ, whenever there is modest over-sampling. The rate-distortion achieved by various combinations of encoders and decoders are summarized in Table \ref{table1}. 

\begin{table}[h]
	\centering
	{\tabulinesep=1.2mm
		\begin{tabu}{ | c | c | c  |}
			\hline	
			&beta-quantization &MSQ \\
			\hline
			TV-min/BLASSO  &$O(M^{1/4}\lambda^{5/4} K^{-\lambda/2})$ &$O(M^{1/2} K^{-1})$ \\
			\hline 
			ESPRIT  & $ O(M^{3/2} \lambda^{1/2} K^{-\lambda})$ & $O(M^{1/2} K^{-1})$ \\
			\hline 
	\end{tabu}}
	\caption{Reconstruction accuracy using either convex or subspace methods from either beta-quantized or memoryless scalar quantized measurements. }
	\label{table1}
\end{table}

\commentout{To put these results in context, we can compare them with an oracle bound. We provide an example of a $\mu\in\calM(\T,S)$ such that if we are given the support of $\mu$ for free, but its amplitudes must be recovered from its quantized Fourier coefficients, then the best possible reconstruction accuracy is $O(K^{-8\lambda})$. There is a significant gap between this oracle lower bound and the ones we establish, which is to be expected because the support estimation step is the most challenging part in super-resolution. }

\commentout{
More precisely, if the (integer) oversampling ratio $\lambda$ is such that 
\[
\lfloor M/\lambda\rfloor - 1\geq 4/\Delta,
\] 
then for any number $K$ of quantization levels, our quantization method guarantees reconstruction accuracy of order $O(\lambda^{3/2} K^{-\lambda/2})$, up to constants which are independent of $K$ and $\lambda$.} 

The remainder of this paper is organized as follows. Section \ref{sec:quan} is the core of this paper. It introduces necessary background on noise-shaping quantization, defines our novel quantizer, and outlines our general approach. Sections \ref{sec:convex} and \ref{sec:subspace} can be read independently from each other, and they combine our quantizer with decoders derived from either on convex optimization (TV-min/BLASSO) or subspace methods (ESPRIT), respectively. The main result for TV-min/BLASSO is presented in Theorem \ref{thm:quanconvex}, while the main result for ESPRIT is given in Theorem \ref{thm:quanesprit}. Section \ref{sec:principle} discusses the main ingredients for combining our quantizer with other robust super-resolution algorithms beyond the ones considered in this paper. Section \ref{sec:num} provides numerical simulations that validate our theorems and show that our method is successful even in the extreme scenario $K=2$. In Section \ref{sec:lower}, we provide lower bounds for MSQ. Theorem \ref{thm:MSQ2} shows that reconstruction from memoryless scalar quantized measurements using Cartesian quantization alphabets cannot achieve rate-distortion better than $O(M^{-1}K^{-1})$, regardless of which algorithm performs the reconstruction. 

To our best knowledge, this paper and its accompanying conference proceeding \cite{gunturk2019high} are the first to present an effective quantizer for spectral super-resolution that takes advantage of redundancy. Designing an effective quantization scheme has important practical implications since it is generally impossible to exactly recover the measure's support and amplitudes whenever there is any amount of noise, even due to small rounding errors. From a mathematical perspective, this paper develops new quantization techniques that we envision can be applied to other measurement systems and recovery methods.

\section{Proposed quantization method}
\label{sec:quan}

\subsection{Basic notation}
$\calM(\T)$ denotes the set of bounded discrete measures on $\T:=[0,1)$ and $\|\cdot\|_{\TV}$ is the total variation norm. The support of a measure $\mu$ is denoted $\supp(\mu)$. Recall that the minimum separation of a discrete measure is denoted $\Delta(\mu)$ and is defined in \eqref{eq:minsep}. For $1\leq p,q\leq\infty$, we let $\|\cdot\|_p$ be the $\ell^p$ norm of a vector and $\|\cdot\|_{p\to q}$ be the $\ell^p$ to $\ell^q$ operator norm. For any matrix $A$, we let $\sigma_k(A)$ denote the $k$-th largest singular value, $A^\dagger$ be the Moore-Penrose pseudo-inverse, and $A^t$ be its transpose. For any vector $u$, we let $u_j$ be its $j$-th entry in the standard coordinate system. Generally, we use $C_1,\dots, C_k$ to denote universal constants whose value may change throughout this paper.

\subsection{Review of noise-shaping for finite frames} 

In this section, we summarize some main ideas behind noise-shaping quantization of finite frame coefficients. Suppose $x\in\C^S$ is an unknown vector and $F\in\C^{M\times S}$ is a known frame, i.e. $F$ is injective, and we would like to quantize the frame coefficients $Fx\in\C^M$ to a vector $q\in\calA^M$ for which we can find an appropriate left inverse that yields good reconstruction.
\begin{enumerate}[(a)]
	\item 
	(Existence of a particular quantizer). Suppose the quantization is done so that there is a matrix $H\in\C^{M\times M}$ and vector $u\in\C^M$ such that $q-Fx=Hu$ and $\|u\|_\infty$ is bounded in dependent of $M$. This form for the quantizer can be traced by to earlier papers on noise-shaping quantization (e.g. \cite{gunturk2003one,gunturk2004approximating}). 
	\item 
	(Existence of a particular inverse). Since $F$ is a frame, there is no unique left inverse of $F$, and the number of choices is an increasing function of the over-sampling rate $M/S$. Suppose $V\in\C^{m\times M}$ is some matrix such that $VF$ is a frame for $\C^S$. Then $(VF)^\dagger V$ is a left inverse of $F$. This type of reconstruction was employed by previous papers (e.g. \cite{gunturk2013sobolev,chou2016distributed,blum2010sobolev}).	 
\end{enumerate}

The $\ell^2$ error of the reconstruction of $x$ via $(VF)^\dagger Vq$ is
\begin{equation}
	\label{eq:frame}
\|x-(VF)^\dagger Vq\|_2
=\|(VF)^\dagger V(Fx-q)\|_2
\leq \frac{\|VH\|_{\infty\to 2}\|u\|_\infty}{\sigma_S(VF)}.
\end{equation}
With this observation at hand, the strategy is clear: construct $H$ and $V$ so that the right hand side is small. 


\subsection{Abstract nonlinear noise-shaping quantization} 
\label{sec:abstract}

There are several aspects of the super-resolution quantization problem that are different from that of the finite frame one. First, although $y=\calF_M\mu$ is a linear combination of vectors $\{(e^{-2\pi ij t_k})_{j=1}^{M-1}\}_{k=1}^{S}$, the $\{t_k\}_{k=1}^S$ are unknown and depend on the measure $\mu$, which means we do not have access to the underlying frame matrix. Second, popular super-resolution algorithms are non-linear, which is in contrast to the finite frame setting where there exist multiple linear left inverses. 

For the aforementioned reasons, we must develop a more general noise-shaping quantization method. Since our approach can be applied to a multitude of problems, we explain our ideas in an abstract setting. Consider any pair $(\Phi,\Psi)$ where $\Phi\colon X\to \C^M$, $\Psi\colon \C^M\to X$, and $\Psi\circ \Phi=I$. It is important to mention that both $\Phi$ and $\Psi$ are allowed to be non-linear maps. In the context of super-resolution, $X$ is a suitable class of measures, $\Phi$ plays the role of $\calF_M$, and $\Psi$ is a non-linear reconstruction algorithm based either on convex optimization or subspace methods. 

Returning back to the abstract setting, $y=\Phi(x)$ represents the analog measurement vector and we must select an alphabet $\calA$ and quantizer $Q\colon \C^M\to\calA^M$ such that reconstruction from quantized measurements $q=Q(y)=Q(\Phi(x))$ remains as close to $x$ as possible uniformly over $x\in X$. Our proposed strategy relies on two crucial assumptions. 

\begin{enumerate}[(a)]
	\item 
	(Existence of a particular quantizer). Suppose there is a linear map $H\colon \C^M\to\C^M$ such that for all $x\in X$, there exist $q\in\calA$ and bounded $u\in \C^M$ satisfying the equation 
	\begin{equation}
	\label{eq:abstract1}
	y-q
	=\Phi(x)-q
	= Hu. 
	\end{equation}
	This defines a mapping $Q\colon \C^M\to\calA^M$ given by $\Phi(x)\mapsto q$.
	
	\item
	(Existence of a particular inverse). Let $\calE\colon X\times X\to\R$ be a distortion function that quantifies the error between two elements of $X$. Assume that there is a linear map $V\colon \C^M\to \C^m$, where $m\leq M$, such that $\Phi_V:=V\circ \Phi$ admits a robust left inverse $\Psi_V$ in the following sense: $\Psi_V\circ \Phi_V=I$ and there exist $C>0$ and $\alpha>0$ such that,
	\begin{equation}
	\label{eq:abstract2}
	\text{for all } x\in X \text{ and } \xi\in \C^m, \quad
	\calE(x,\Psi_V(\xi))
	\leq C\|\Phi_V(x)-\xi\|_2^\alpha. 
	\end{equation}
\end{enumerate}

The implicit constant $C$ in \eqref{eq:abstract2} serves the same role as $1/\sigma_S(VF)$ in the frame setting in inequality \eqref{eq:frame}. Under these two assumptions, we can easily establish that $\Psi_V\circ V$ is a robust decoder for $\Phi$. Indeed, by \eqref{eq:abstract1}, we have
\[
\Phi_V(x)-Vq=V\Phi(x)-Vq=V(y-q)=VHu,
\]
and according to \eqref{eq:abstract2}, we have 
\[
\calE(x,\Psi_V(Vq))
\leq C\|VHu\|_2^\alpha
\leq C \|VH\|_{\infty\to 2}^\alpha\|u\|_\infty^\alpha.  
\]
Thus, the rate-distortion for this scheme is largely controlled by $\|VH\|_{\infty\to 2}$. We will carry out an explicit construction for super-resolution in the next subsection.

\subsection{Beta-encoding for super-resolution} 
\label{sec:beta}

In this section, we formally define our proposed quantizer for super-resolution. Recall that the number of Fourier measurements $M$ and parameter $\Delta>0$ are fixed, and we assume there exists $\lambda\geq 1$ that satisfies inequality \eqref{eq:lambda}. Without loss of generality, we assume that $M$ is divisible by $\lambda$ and set $m := M/\lambda$ because if $M$ is not divisible by $\lambda$, then we simply discard the extra samples. 

We introduce a parameter $\beta>1$ that will be chosen later. Letting $I_m$ denote the $m\times m$ identity matrix, consider the $m\times M$ matrix $V:=V_{\beta}$ where
\begin{equation}
\label{eq:V}
V:=
\begin{bmatrix}
I_m &\beta^{-1}I_m &\cdots &\beta^{-\lambda+1} I_m
\end{bmatrix}. 
\end{equation}

One of the key properties that we exploit throughout this paper is the following algebraic identity (which holds for all measures, not necessarily atomic) specialized to an atomic $\mu$. For each $\ell=1,\dots,m$, we have
\begin{gather}
\label{eq:alg}
\begin{split}
(V\calF_M\mu)_\ell
&=\sum_{k=0}^{\lambda-1} \beta^{-k}\, \hat\mu(mk+\ell) \\
&=\sum_{k=0}^{\lambda-1} \beta^{-k} \int_0^1 e^{-2\pi i(mk+\ell)t} \ d\mu(t)
=\int_0^1 \( \sum_{k=0}^{\lambda-1} \beta^{-k} e^{-2\pi i mk t}\) e^{-2\pi i \ell t} \ d\mu(t). 
\end{split}
\end{gather}
This calculation motivates us to define the functions 
\begin{equation*}
	w(t)
	:=\sum_{k=0}^{\lambda-1} \beta^{-k} e^{-2\pi i k t}
	=\frac{1-\beta^{-\lambda}e^{-2\pi i \lambda t}}{1-\beta^{-1}e^{-2\pi i t}}
	\quad\text{and}\quad
	w_m(t):=w(mt). 
\end{equation*}
Note that $w_m$ is uniformly bounded above and below,
\begin{equation}
	\label{eq:weights}
	\frac{1}{C_\beta}
	\leq |w_m(t)|
	\leq C_\beta, 
	\quad\text{where}\quad
	C_\beta := \frac{1+\beta^{-1}}{1-\beta^{-1}}. 
\end{equation}
Letting $L_{\beta,\lambda}$ stand for the Lipschitz constant of $w(t)$, it can be shown that 
\begin{equation*}
	L_{\beta,\lambda}\leq \frac{4\pi\lambda\beta}{(\beta-1)^2}
	\quad\text{and}\quad 
	| w_m^{-1}(s) - w_m^{-1}(t) | \leq L_{\beta,\lambda}m\,  |s-t|_\T.
\end{equation*}

Following terminology in \cite{chou2016distributed}, we call $V\calF_M\mu\in\C^m$ the {\it condensed measurements}. With this notation at hand, the calculation \eqref{eq:alg} can be summarized as
\begin{equation*}
V\calF_M\mu = \calF_m (w_m\mu). 
\end{equation*}
Observe that $\mu$ and $w_m\mu$ have identical supports, but different amplitudes. 

Let us define $H:=H_\beta$ to be the $M\times M$ matrix where 
\begin{equation}
\label{eq:H}
H:=
\begin{bmatrix}
I_m \\
-\beta I_m &\ddots \\
&\ddots &I_m \\
& & -\beta I_m &I_m
\end{bmatrix}. 
\end{equation}
The following is a special case of \cite[Lemma 2]{chou2017distributed}:

\commentout{
\begin{equation}
H_{j,k}:=
\begin{cases} 
1, & \mbox{if }j=k, \\
-\beta, & \mbox{if }j = k{+}m \mbox{ and } 1 \leq k \leq M{-}m. 
\end{cases}
\end{equation}
}

\begin{lemma}
	Let $K\geq 2$ be an integer, and suppose the parameters $A,\beta,\delta>0$ satisfy the inequality
	\begin{equation}
	\label{eq:parameter}
	\beta + A\delta^{-1} \leq K, 
	\end{equation}
	and consider the quantization alphabet 
	\begin{equation}
	\label{eq:alphabetdelta}
	\calA 
	:= \calA_K
	:= \delta (\calZ_K+i\calZ_K),
	\end{equation}
	where $\calZ_K$ denotes the $K$-term origin-symmetric arithmetic progression of integers with spacing $2$, namely
	$$
	\calZ_K
	:=\{ -K+1, -K+3, \dots, K-3, K-1\}.
	$$ 
	Then, for any $y\in\C^M$ with $\|y\|_\infty\leq A$, there exists $q\in\calA^M$ and $u\in\C^M$ with $\|u\|_\infty\leq \sqrt{2} \delta$ satisfying the relationship
	\begin{equation*}
	y-q=H u.
	\end{equation*}
\end{lemma}

We note that for any $1< \beta < K$ and any $0 \leq A < \infty$, there exists $\delta > 0$ such that the condition \eqref{eq:parameter} is satisfied. Let $\calA$ be the alphabet and $Q$ denote the mapping $y \mapsto q$ implied by the above lemma, which can be implemented by means of a simple recursive algorithm \cite{chou2016distributed, chou2017distributed}. Indeed, letting $\text{round}_\calA$ be the rounding operation to $\calA$, we recursively define
\begin{align*}
	\begin{cases}
		\ q_k := {\rm round}_\calA(y_k) &\text{if } k\leq m, \\
		\ u_k := y_k-q_k &\text{if } k\leq m, \\
		\ q_k:= {\rm round}_\calA(y_k+\beta u_{k-m})  &\text{if } k>m, \\
		\ u_k:= y_k-q_k+\beta u_{k-m} &\text{if } k>m. 
	\end{cases}
\end{align*}

The significance of $V$ and $H$ is that $VH$ is very small when $\beta$ or $\lambda$ is large. Indeed, as shown in \cite{chou2016distributed}, we have
$\|VH\|_{\infty\to 2} = \sqrt{m}\ \beta^{-\lambda+1}$. The above findings provide the core strategy of our proposed quantization and recovery method. For any $\mu\in\calM(\T)$ whose total variation is bounded above by $A$, if $q:=Q(\calF_M\mu)$, then $Vq$ is now a small perturbation of $Vy=V\calF_M\mu=\calF_m (w_m \mu)$ because 
\begin{equation*}
	\|Vy-Vq\|_2
	\leq \|VH\|_{\infty\to 2}\|u\|_\infty
	\leq \sqrt{2m}\ \beta^{-\lambda+1} \delta
	=:\epsilon_V.
\end{equation*}
Interpreting $Vq$ as the first $m$ noisy Fourier coefficients of $w_m \mu$, under appropriate conditions, we can use any robust super-resolution recovery algorithm $\Psi_m\colon \C^m\to \calM(\T)$ to compute a discrete measure 
\[
\nu 
:=\Psi_m(Vq)
:=\sum_{k=1}^{\tilde S} \tilde b_k \delta_{\tilde t_k} 
\quad\text{and}\quad 
\tilde T:=\{t_k\}_{k=1}^{\tilde S}, 
\]
where it is possible that $\tilde S\not=S$. A key property of our approach is that $\mu$ and $w_m \mu$ have the same support. We think of $\tilde T$ as a perturbation of $T$ and we define the measure
\begin{equation*}
\tilde \mu
:=w_m^{-1}\nu 
=\sum_{k=1}^{\tilde S} \frac{\tilde b_k}{w_m(\tilde t_k)} \delta_{\tilde t_k}. 
\end{equation*}
It is important to emphasize that this re-weighting procedure depends on the quantizer parameters and the output of $\Psi_m$, but does not require any information about the amplitudes or support of $\mu$.

Before we continue, let us bound some of the constants that we introduced, and these estimates will frequently appear in our subsequent analysis. We typically view the parameters $A$, $\Delta$ and $\lambda$ as fixed, while the number of levels $K$ is variable. The other parameters $1<\beta< K$ and $\delta>0$ are free for us to select, so we optimize for these in terms of $\lambda,A,K$. At this point, we note the following elementary fact: For any $K\geq 2$, setting 
\begin{equation}
\label{eq:parameters}
\beta:= \frac{K(\lambda+1)}{\lambda+2} 
\quad\text{and}\quad
\delta:= \frac{(\lambda+2)A}{K}, 
\end{equation}
results in 
\[
\beta + A\delta^{-1} = K
\quad\text{and}\quad
\delta \beta^{-\lambda+1} < \mathrm{e} A (\lambda+1) K^{-\lambda}.
\]
(See, e.g. \cite[Lemma 3.2]{chou2016distributed} and \cite[Lemma 1]{chou2017distributed}.) This choice of parameters results in 
\begin{equation}
\label{eq:parameters3} 
\beta \geq \frac{4}{3}, \quad
C_\beta \leq 7, \quad 
L_{\beta,\lambda} 
\leq 12\pi \lambda, 
\quad
\epsilon_V \leq \mathrm{e}A (2M)^{1/2} \lambda^{-1/2} (\lambda+1) K^{-\lambda}.
\end{equation}

\medskip 
\fbox{
\begin{minipage}{\textwidth}
	\par \noindent {\bf System parameters and assumptions:}
	\begin{itemize}\itemsep-0.25em
		\item $M$ (number of Fourier measurements),
		\item $\lambda \geq 1$ (oversampling ratio) such that $M=\lambda m$, 
		\item $A$ (upper bound on TV-norm of the input measures),
		\item $K$ (number of quantization levels),
		\item $\beta$ and $\delta$ defined via \eqref{eq:parameters}.
	\end{itemize}
	
	\par \noindent {\bf Encoding (quantization) stage:}
	\begin{itemize}\itemsep-0.25em
		\item Input to quantizer: $y=\calF_M\mu$ such that $\|y\|_\infty \leq A$,
		\item $V$ and $H$ defined via \eqref{eq:V} and \eqref{eq:H},
		\item Quantization alphabet: $\calA:= \delta(\calZ_K + i \calZ_K)$,
		\item Output of quantizer: $q\in \calA^M$ such that $\|Vy - Vq\|_2 \leq \epsilon_V$.
	\end{itemize}
	
	\par \noindent {\bf Decoding (recovery) stage:}
	\begin{itemize}\itemsep-0.25em
		\item Input to decoder: $q$,
		\item 
		Feed $Vq$ into an any super-resolution algorithm to compute a measure 
		\[
		\nu=\sum_{k=1}^{\tilde S} \tilde b_k \delta_{\tilde t_k},
		\]
		and abort if it cannot be found (e.g. invalid measurements),
		\item Output of decoder: 
		\[
		\tilde \mu = \sum_{k=1}^{\tilde S}  \frac{\tilde b_k}{w_m(\tilde t_k)} \delta_{\tilde t_k} .
		\]
	\end{itemize}	
\end{minipage}
}

\section{Quantization for convex methods}
\label{sec:convex}

\subsection{Preliminaries}

There are a number of robust convex algorithms for super-resolution. These are based upon finding a measure that explains the observations and has minimal total variation. In the noiseless setting where we have access to the clean Fourier measurements $\calF_M\mu$, total variation minimization (TV-min) produces an estimate of $\mu$ by finding 
\begin{equation*}
TV_{M,0}(\calF_M\mu):=\argmin\{\|\nu\|_\TV\colon \calF_M\mu=\calF_M\nu\}. 
\end{equation*}
In \cite{candes2014towards} it was shown that $M\geq 257$ and $M-1\geq 4/\Delta(\mu)$ are sufficient conditions\footnote{Some super-resolution papers assume that the Fourier samples are indexed by $\{-N,\dots,N\}$ for some integer $N$, whereas we assume the frequencies lie in $\{0,\dots,M-1\}$. Both settings are equivalent, but we need to be careful with the transcription process: our number of measurements is $M$, which corresponds to $2N+1$ in some other papers.} to guarantee that $\mu=TV_{M,0}(\calF_M\mu)$. The numerical constant $4$ that appears in the separation assumption can be reduced at the price of increasing the number of measurements $M$ (e.g. \cite{fernandez2016super}). On the other hand, a minimum separation assumption on the order of $1/M$ is necessary for TV-min: there exists a sufficiently small $C$ such that for any $M$, there exists is a $\mu\in\calM(\T)$ with $\Delta(\mu)\leq C/M$ such that $\mu$ is not a solution to TV-min given noiseless Fourier data $y=\calF_M\mu$ (e.g. \cite{duval2015exact,benedetto2020super}). 

In the presence of noise, given noisy data $y\in\C^{M}$ such that $\|\calF_M\mu-y\|_2\leq \epsilon$ for a known $\epsilon>0$, the TV-min algorithm outputs an estimate of $\mu$ given by 
\begin{equation*} 
TV_{M,\epsilon}(y)
:=\argmin_{\nu\in\calM(\T)} \, \Big\{\|\nu \|_{\TV}\colon 
\|y-\calF_M \nu\|_2 \leq \epsilon\Big\}.
\end{equation*}
This is a convex program whose feasibility is guaranteed by the assumption $\|y - \calF_M\mu\|_2 \leq \epsilon$. The solution may not be unique, but it is known that there is at least one minimizer which is a discrete measure (e.g. \cite{azais2015spike,fernandez2013support}). Sometimes TV-min is recast in unconstrained form. For a fixed parameter $\tau>0$, the BLASSO algorithm estimates $\mu$ by computing the solution to 
\begin{equation*}
BL_{M,\tau}(y)
:=\argmin_{\nu\in\calM(\T)}\Big\{\tau \|\nu\|_{\TV}+\frac{1}{2} \|y-\calF_M\nu\|^2_2\Big\}. 
\end{equation*}

Under identical separation assumptions, it can be shown that solutions to either TV-min or BLASSO approximate the target measure in an appropriate sense. For any discrete measures $\rho=\sum_{j=1}^R u_j\delta_{r_j}$ and $\nu=\sum_{k=1}^S v_k \delta_{s_k}$, which may contain different numbers of atoms, we define the ``neighborhood'' index sets\footnote{The numerical constant $0.3298$ that appears in \eqref{eq:Ij} is double the one found in \cite{fernandez2013support}, again due to our convention that the Fourier samples are from $\{0,1,\dots,M-1\}$.}
\begin{equation}\label{eq:Ij}
\calI_j^M
:=\calIM_j(\rho,\nu):=\big \{k: |r_j-s_k|_\T \leq 0.3298\,(M-1)^{-1} \big \}, \quad\text{for}\quad  j=1,\dots,R,
\end{equation}
and the residual index set 
\begin{equation*}
\calIM_0
:=\calIM_0(\rho,\nu) := \{1,\dots,S\}\setminus \bigcup_{j=1}^R \ \calIM_j(\rho,\nu).
\end{equation*}
Note that $\calIM_j(\rho,\nu)$ and $\calIM_k(\rho,\nu)$ are disjoint whenever $\Delta(\rho)\geq 1/M$. To quantify the error due to approximating $\rho$ by $\nu$, we define the following:
\begin{align*}
\calE_1^M(\rho,\nu)
&:=\max_j \Big| u_j -\sum_{k \in \calIM_j} v_k \Big|, \\
\calE_2^M(\rho,\nu)
&:=\sum_j \sum_{k \in \calIM_j} |v_k|\,|r_j - s_k|_\T^2, \\
\calE_3^M(\rho,\nu)
&:=\sum_{k \in \calIM_0} |v_k |.
\end{align*}
The first, second and third terms capture, respectively, the amplitude error, support localization error, and total variation of the artificial atoms. Note that all three distortion functions are not symmetric in $\rho$ and $\nu$.

We recall the following error bounds for TV-min \cite[Theorem 1.2]{fernandez2013support} and for BLASSO \cite[Theorems 2.1 and 2.2]{azais2015spike} under the assumption that the minimum separation is sufficiently large. 

\begin{theorem}
	\label{thm:convex}
	There exist universal constants $C_1,C_2,C_3>0$ such that the following hold. For any integer $M\geq 257$, $\Delta\geq 4/(M-1)$, and $\mu\in\calM(\T,\Delta)$, given the Fourier measurements $y=\calF_M\mu+z$ with $\|z\|_2\leq\epsilon$ for some known $\epsilon>0$, for any discrete $\tilde\mu:=TV_\epsilon(y)$ or $\tilde\mu:=BL_\epsilon(y)$, we have
	\begin{equation*}
	\calE_1^M(\mu,\tilde\mu) \leq C_1 \epsilon, \quad
	\calE_2^M(\mu,\tilde\mu) \leq C_2 M^{-2}\epsilon, \quad 
	\calE_3^M(\mu,\tilde\mu) \leq C_3 \epsilon.
	\end{equation*}
\end{theorem}

An immediate consequence of this robustness result is an upper bound for the reconstruction error for when $\calF_M\mu$ is quantized to the alphabet $\calA_K$ by MSQ. Indeed, if $z$ is quantization error, then $\|z\|_2\leq \sqrt{M}\|z\|_\infty$, and this theorem shows that the overall reconstruction accuracy from MSQ samples using either TV-min or BLASSO is $O(\sqrt{M} K^{-1})$.

\subsection{Noise-shaping quantization for TV-min and BLASSO}

We show how to combine our proposed quantization method with TV-min and BLASSO to derive an encoding-decoding scheme that offers a significantly better reconstruction accuracy than MSQ. To do this, let $\|\cdot\|_\Lip$ be a Lipschitz norm on $\T$, which we define to be
\[
\|f\|_\Lip
:= \max(\|f\|_\infty, |f|_{\Lip})
\quad\text{and}\quad
|f|_\Lip
:= \sup_{x,y\in\T} \frac{|f(x)-f(y)|}{|x-y|_\T}.
\]
We let $\|\cdot\|_{\Lip^*}$ denote its dual norm. Since each $\rho\in \calM(\T)$ is a bounded linear functional on $\Lip(\T)$, we define a distortion function,
\[
\calE_\Lip(\rho,\nu)
:=\|\rho-\nu\|_{\Lip^*}
=\sup_{\|\varphi\|_\Lip \leq 1} \Big| \int_0^1 \varphi \ d(\rho-\nu) \Big|.
\]
If $\rho$ and $\nu$ are both probability measures, then $\calE_\Lip$ is the usual $p=1$ Wasserstein metric, but for super-resolution, we in general must deal with complex measures that do not necessarily have the same total variation. 

The following lemma allows us to connect $\calE_\Lip$ with $\calE^M_1$, $\calE^M_2$, and $\calE^M_3$. This also illustrates why $\calE_\Lip$ is useful for super-resolution. 

\begin{lemma}
	\label{lem:Elip1}
	For any $\rho\in\calM(\T,\Delta)$ with $R$ atoms,  atomic $\nu\in\calM(\T)$, and integer $M\geq \Delta^{-1}$, we have
	\[
	\calE_\Lip(\rho,\nu)
	\leq R \calE_1^M(\rho,\nu) + \|\nu\|_{\TV}^{1/2} \big( \calE_2^M(\rho,\nu) \big)^{1/2} + \calE_3^M(\rho,\nu).
	\]
\end{lemma}

\begin{proof}
	By assumption, we can represent $\rho$ and $\nu$ as $\rho=\sum_{j=1}^R u_j\delta_{r_j}$ and $\nu=\sum_{k=1}^S v_k \delta_{s_k}$. For any $\epsilon>0$, there exists $\varphi$ with $\|\varphi\|_\Lip\leq 1$ such that
	\[
	\calE_\Lip(\rho,\nu)
	\leq \Big| \int_0^1 \varphi \ d(\rho-\nu) \Big|+\epsilon. 
	\]
	By definition, we have
	\begin{align*}
	\int_0^1 \varphi \ d(\rho-\nu) 
	&=\sum_{j=1}^R u_j \varphi(r_j) - \sum_{j=1}^R\sum_{k\in \calI_j^M} v_k \varphi(s_k) - \sum_{k\in\calI_0^M} v_k\varphi(s_k) \\
	&=\sum_{j=1}^R \(u_j-\sum_{k\in\calI_j^M} v_k \) \varphi(r_j) + \sum_{j=1}^R \sum_{k\in \calI_j^M} v_k \big(\varphi(r_j)-\varphi(s_k)\big) - \sum_{k\in\calI_0^M} v_k\varphi(s_k). 
	\end{align*}
	It follows that 
	\begin{align*}
	\Big| \int_0^1 \varphi \ d(\rho-\nu) \Big|
	&\leq \sum_{j=1}^R \Big| u_j-\sum_{k\in\calI_j^M} v_k \Big| + \sum_{j=1}^R \sum_{k\in \calI_j^M} |v_k| |r_j-s_k|_{\T} + \sum_{k\in\calI_0^M} |v_k| \\
	&\leq R \calE_1^M(\rho,\nu) + \sum_{j=1}^R \|\nu\|_{\TV}^{1/2} \(\sum_{k\in \calI_j^M} |v_k| |r_j-s_k|_{\T}^2\)^{1/2} + \calE_3^M(\rho,\nu) \\
	&= R \calE_1^M(\rho,\nu) + \|\nu\|_{\TV}^{1/2} \big( \calE_2^M(\rho,\nu) \big)^{1/2} + \calE_3^M(\rho,\nu). 
	\end{align*} 
	Since $\epsilon$ is arbitrary, this completes the proof. 
\end{proof}

Our distortion $\calE_\Lip$ upper bounds other distortion functions as well, such as the $L^p(\T)$ norm of the convolution of $\rho-\nu$ with a suitable kernel, which was employed in \cite{candes2013super,li2017elementary}. Indeed, for any Lipschitz $\psi$, we have
\begin{align*}
\|\psi*(\rho-\nu)\|_{\infty}
&\leq \|\psi\|_\Lip \Big\| \frac{\psi}{\|\psi\|_{\Lip}} *(\rho-\nu)\Big\|_\infty \\
&\leq \|\psi\|_\Lip \|\rho-\nu\|_{\Lip^*} \\
&= \|\psi\|_\Lip \, \calE_\Lip(\rho,\nu).
\end{align*}

The next lemma establishes how multiplication of atomic measures by any Lipschitz function that is bounded above and below affects $\calE_\Lip$. 


\begin{lemma}
	\label{lem:Elip2}
	Fix any Lipschitz $\psi$ for which there exists $C\geq 1$ such that $C^{-1}\leq |\psi(t)|\leq C$ for all $t\in\T$. For any $\rho\in\calM(\T)$, we have 
	\begin{align*}
	\min\big( C^{-1}, 2\|\psi\|_\Lip^{-1} \big) \|\psi\rho\|_{\Lip^*}
	\leq 
	\|\rho\|_{\Lip^*}
	\leq \max\big(C,C^2|\psi|_\Lip \big) \, \|\psi\rho\|_{\Lip^*}.
	\end{align*}
\end{lemma}

\begin{proof}
	We concentrate on the second claimed inequality first. For any $\epsilon>0$, there exists a Lipschitz $\varphi$ such that $\|\varphi\|_\Lip\leq 1$ and 
	\[
	\|\rho\|_{\Lip^*}
	\leq \Big| \int_0^1 \varphi \ d\rho\Big| + \epsilon
	\leq \Big\| \frac{\varphi}{\psi} \Big\|_\Lip \|\psi\rho\|_{\Lip^*}+\epsilon. 
	\]
	Note that $\|\varphi/\psi\|_\infty \leq C\|\varphi\|_\infty\leq C$ and that  $| \varphi/\psi |_\Lip\leq C^2 |\psi|_{\Lip}$, which proves the second claimed inequality. 
	
	For the first claimed inequality, for any $\epsilon>0$, there exists a Lipschitz $\varphi$ such that $\|\varphi\|_\Lip\leq 1$ and 
	\[
	\|\psi\rho\|_{\Lip^*}
	\leq \Big| \int_0^1 \varphi \psi \ d\rho \Big| + \epsilon
	\leq \| \varphi \psi \|_\Lip \, \|\rho\|_{\Lip^*}+\epsilon. 
	\]
	The conclusion follows once we use the observation that $\| \varphi \psi \|_\infty \leq C$ and 
	\[
	| \varphi \psi |_\Lip
	\leq |\varphi|_{\Lip} \|\psi\|_\infty + \|\varphi\|_\infty |\psi|_\Lip
	\leq \|\psi\|_\infty + |\psi|_{\Lip}
	\leq 2\|\psi\|_{\Lip}. 
	\]

\end{proof}

The proceeding two lemmas will enable us to greatly simplify our presentation of the following theorem which quantifies the performance of our quantization and recovery strategy. 

\begin{theorem}
	\label{thm:quanconvex}
	There exists a universal constant $C>0$ such that the following holds. For any $A>0$ and any integers $K\geq 2$, $\lambda\geq 1$, $m\geq 257$, let $M=\lambda m$ and $(Q,D)$ be the quantizer-decoder pair summarized in Section \ref{sec:quan} that uses either $TV_{m,\epsilon_V}$ or $BL_{m,\epsilon_V}$ as the super-resolution algorithm in the decoding step.
	
	For any $\Delta\geq 4/(m-1)$ and $\mu\in \calM(\T,\Delta)$ with $\|\mu\|_{\TV}\leq A$, given the quantized Fourier measurements $q:=Q(\calF_M\mu)$, there is an output $\tilde \mu:= D(q)$ that satisfies the error bound
	\begin{equation*}
	\calE_\Lip(\mu,\tilde\mu)
	\leq C A S M^{3/2} K^{-\lambda} + C A M^{1/4} \lambda^{5/4} K^{-\lambda/2}.
	\end{equation*}
\end{theorem}

\begin{proof}
	We start by recalling some essential properties our encoding-decoding scheme. Let $\mu=\sum_{j=1}^S a_j\delta_{t_j}\in \calM(\T,\Delta)$. We have $\|\calF_M\mu\|_\infty\leq \|\mu\|_{\TV}\leq A$ and $V\calF_M\mu=\calF_m(w_m \mu)$. Letting $q:=Q(\calF_M \mu)\in\C^M$, we define the condensed quantization error $z:=Vy-Vq$ and recall that $\|z\|_2\leq \epsilon_V$. 
	
	If we use TV-min as the super-resolution algorithm in the decoding step, existence of $\nu:=TV_{m,\epsilon_V}(Vq)$ follows from the observation that $w_m\mu$ is a feasible measure for $TV_{m,\epsilon_V}$ since $\|z\|_2\leq \epsilon_V$. If we instead use BLASSO, it always admits a solution $\nu:=BL_{m,\epsilon_V}(Vq)$. We next bound $\|\nu\|_{\TV}$, which we split into two cases.
	\begin{enumerate}[(a)]
		\item
		TV-min. Since $\nu$ is a minimizer of TV-min and $w_m\mu$ is also feasible, we have 
		\[
		\|\nu\|_{\TV} 
		\leq \|w_m\mu\|_{\TV} 
		\leq \|w\|_\infty \|\mu\|_{\TV}
		\leq C_\beta A. 
		\]
		\item 
		BLASSO. Since $\nu$ is a minimizer of BLASSO, we have
		\begin{align*}
			\|\nu\|_{\TV}
			&\leq \|\nu\|_{\TV}+\frac{1}{2\epsilon_V} \|\calF_m \nu-Vq\|_2^2 
			\leq \|w_m\mu\|_{\TV}+\frac{1}{2\epsilon_V} \|\calF_m (w_m\mu)-Vq\|_2^2 \\
			&\leq \|w\|_\infty \|\mu\|_{\TV} +\frac{1}{2\epsilon_V} \|Vy-Vq\|_2^2 \leq C_\beta A+\epsilon_V.	 
		\end{align*}
	\end{enumerate}
	From here onward, let $\nu$ be the output of either these algorithms, and we have the upper bound, 
	\begin{equation}
	\label{eq:convex0}
	\|\nu\|_{\TV}
	\leq C_\beta A+\epsilon_V.
	\end{equation}

	We can use Theorem \ref{thm:convex} to upper bound $\calE_j^m(w_m\mu,\nu)$ since $w_m\mu\in \calM(\T,\Delta)$ and we assumed that $\Delta\geq 4/(m-1)$. According to Lemma \ref{lem:Elip1} and inequality \eqref{eq:convex0}, we have
	\begin{gather}
		\label{eq:convex1}
		\begin{split}
		\calE_\Lip(w_m\mu,\nu)
		&\leq S \calE_1^m(w_m\mu,\nu) + \|\nu\|_{\TV}^{1/2}  \big( \calE_2^m(w_m\mu,\nu) \big)^{1/2} + \calE_3^m(w_m\mu,\nu) \\
		&\leq C_1 S \epsilon_V + C_2^{1/2} m^{-1} (C_\beta A+\epsilon_V)^{1/2} \epsilon_V^{1/2} + C_3 \epsilon_V \\
		&\leq \big(C_1S +C_2^{1/2} m^{-1}  + C_3\big) \epsilon_V + C_2^{1/2} C_\beta^{1/2} m^{-1} A^{1/2} \epsilon_V^{1/2}, 	
		\end{split}		
	\end{gather}
	where $C_1,C_2,C_3$ are the universal constants in Theorem \ref{thm:convex}. 
	
	Since $\tilde\mu:=D(q)=w_m^{-1}\nu$, applying Lemma \ref{lem:Elip2} with $\psi=w_m$ and combining it with \eqref{eq:convex1} and the choice of parameters \eqref{eq:parameters3}, we have 
	\begin{align*}
	\calE_\Lip(\mu,\tilde\mu)
	&\leq C_\beta^2 |w_m|_\Lip \, \calE_\Lip(w_m\mu,w_m\nu) \\
	&\leq C_\beta^2 m L_{\beta,\lambda} \big( C_1S +C_2^{1/2} m^{-1}  + C_3 \big) \epsilon_V + C_2^{1/2} C_\beta^{5/2} L_{\beta,\lambda} A^{1/2} S \epsilon_V^{1/2}  \\
	&\leq C A S M^{3/2} \lambda^{1/2} K^{-\lambda} + C A M^{1/4} \lambda^{5/4} K^{-\lambda/2}.
	\end{align*}
	
\end{proof}

This theorem quantifies the error in terms of $\calE_\Lip$ primarily for mathematical elegance and convenience. Indeed, under the same conditions and for the same decoder-encoder described in Theorem \ref{thm:convex}, we have the identical error estimate in terms of $\calE_j^m$. The following theorem was fully proved in our accompanying proceeding \cite{gunturk2019high}. 

\begin{theorem}
	\label{thm:quanconvex2}
	There exist universal constants $C_1,C_2,C_3>0$ such that the following hold. For any $A>0$ and any integers $K\geq 2$, $\lambda\geq 1$, $m\geq 257$, let $M=\lambda m$ and $(Q,D)$ be the quantizer-decoder pair summarized in Section \ref{sec:quan} that uses either $TV_{m,\epsilon_V}$ or $BL_{m,\epsilon_V}$ as the super-resolution algorithm in the decoding step.
	
	For any $\Delta\geq 4/(m-1)$ and $\mu\in \calM(\T,\Delta)$ with $\|\mu\|_{\TV}\leq A$, given the quantized Fourier measurements $q:=Q(\calF_M\mu)$, there is an output $\tilde \mu:= D(q)$ that satisfies the error bounds
	\begin{align*}
		\calE_1^m(\mu,\tilde\mu)
		& \leq C_1 A M^{1/2}\lambda^{1/2} K^{-\lambda} + C_1 (1+ M^{1/2}\lambda^{1/2} K^{-\lambda})^{1/2} A M^{1/4}\lambda^{5/4} K^{-\lambda/2},  \\
		\calE_2^m(\mu,\tilde\mu) 
		& \leq C_2 A \lambda K^{-\lambda}, \\
		\calE_3^m(\mu,\tilde\mu) & \leq C_3 A M^{1/2} \lambda^{1/2} K^{-\lambda} \label{eq:arec3}.
	\end{align*}
\end{theorem}

\section{Noise-shaping quantization for subspace methods}
\label{sec:subspace}

\subsection{Preliminaries} 

{\it Subspace methods} refer to a collection of algorithms introduced over a quarter century ago. Although this paper focuses on ESPRIT, many aspects of our results can be adapted to other subspace methods such as MUSIC and matrix pencil method. 


These methods exploit algebraic properties of two matrices. For a given integer $M>0$, the Fourier matrix associated with a discrete set $T=\{t_k\}_{k=1}^S$ is defined as
\begin{equation}
\label{eq:vander}
\Phi_M
:=\Phi_M(T)
:=\Big( e^{-2\pi i jt_k}\Big)_{0\leq j\leq M-1, \, 1\leq k\leq S}.
\end{equation}
With this notation in place, if $\mu=\sum_{j=1}^S a_j\delta_{t_j}$, then 
\[
\calF_M\mu=\Phi_M(T)a.
\]
Throughout, we fix an integer $N$ satisfying the inequalities $S\leq N-1\leq M-S$. The Hankel matrix of a vector $u\in \C^M$ is defined as
\[
H_N(u)
:=\Big( u_{j-k} \Big)_{0\leq j\leq N-1, \, 0\leq k\leq M-N}. 
\]
The Hankel matrix of the Fourier data $\calF_M\mu$ enjoys the decomposition
\[
H_N(\calF_M\mu)
=\Phi_N(T) D_a \Phi_{M-N+1}(T)^t,
\quad\text{where}\quad D_a:=\text{diag}(a_1,\dots,a_S). 
\]

Subspace methods such as ESPRIT take noisy Fourier data $y$ corresponding to some measure $\mu$ and the total number of atoms $S$. We do not explain the mechanisms of ESPRIT and refer the reader to  \cite{li2020super} for further details. The algorithm first outputs a discrete set $\tilde T=\{\tilde t_j\}_{j=1}^S$. If the noise energy is sufficiently small, then 
\[
\min_{\pi} \max_{j} |t_j-\tilde t_{\pi(j)}|_\T
\leq \frac{\Delta(\mu)}{2},
\]
where the minimum is taken over all permutations $\pi$ on $\{1,\dots,S\}$. The algorithm next estimates the amplitudes via least squares: $\tilde a:=\Phi_M(\tilde T)^\dagger y$. The ESPRIT decoder $ES_{M,N,S}$ provides us with a discrete measure whose support is $\tilde T$ and its amplitudes $\tilde a$, namely,
\[
ES_{M,N,S}(y)
:=\sum_{j=1}^S \tilde a_{j} \delta_{\tilde t_j}. 
\]

Several papers have derived error bounds for ESPRIT in terms of the smallest amplitude of the measure, noise energy, and $\sigma_S(\Phi_M)$ under appropriate separation assumptions (e.g. \cite{aubel2016deterministic,fannjiang2016compressive,li2020super}). The important feature of these estimates is that the support and amplitude reconstruction errors are linear in the noise energy $\|z\|_2$. 

To give a clean presentation, we define a mixed $\ell^\infty+\ell^2$ error between two measures $\rho=\sum_{j=1}^S u_j\delta_{r_j}$ and $\nu=\sum_{j=1}^S v_j\delta_{s_j}$, where are assumed to have the same number of atoms. We define the best permutation\footnote{When the noise energy is sufficiently small, there is a unique optimal permutation. In the subsequent results, our assumptions guarantee that this is the case.} between their supports as,
\[
\pi^*
:=\argmin_{\pi} \max_j |r_j-s_{\pi(j)}|_\T. 
\]
For this optimal permutation, we define 
\begin{equation}
\label{eq:error2}
\calE_{\infty,2}(\rho,\nu)
:= \max_{j} |r_j-s_{\pi^* (j)}|_\T + \frac{1}{\|u\|_{2}}\( \sum_{j=1}^S \big|u_j-v_{\pi^*(j)}\big|^2 \)^{1/2}.  
\end{equation}
Here, the $\|u\|_2^{-1}$ factor is a natural normalization constant so that both the support error (first term in \eqref{eq:error2}) and amplitude error (second term in \eqref{eq:error2}) are dimensionless quantities. 

We also define the class of measures
\[
\calM(\T,\Delta,A,B,S), 
\]
as the collection of all $\nu\in \calM(\T,\Delta)$ such that $\nu$ has $S$ atoms, $\|\nu\|_{\TV}\leq A$, and the absolute value of its amplitudes are bounded below by $B$. The following theorem combines previously established results in \cite{li2020super,moitra2015matrixpencil}.

\begin{theorem}
	\label{thm:LL}
	For any integers $S\geq 1$, even $M\geq 8S$, and $N:=M/2+1$, and for any $0<B\leq A$ and $\Delta\geq 4/(M-1)$, given noisy Fourier measurements $y=\calF_M\mu+z$ of $\mu\in\calM(\T,\Delta,A,B,S)$ where
	\begin{equation}
	\label{eq:noisecond}
	\|z\|_2\leq \frac{\Delta BM}{1280 S^2},
	\end{equation} 
	the output $\tilde\mu:=ES_{M,N,S}(y)$ has precisely $S$ atoms and satisfies
	\begin{equation*}
	\calE_{\infty,2}(\mu,\tilde\mu)
	\leq 320 B^{-1} S^2 M^{-1} \|z\|_2 + 3400 B^{-1} S^{5/2} \|z\|_2. 
	\end{equation*}
\end{theorem}

\begin{proof}
	
	Let $\mu=\sum_{j=1}^S a_j\delta_{t_j}\in M(\T,\Delta,A,B,S)$. We first note that $\|H_N(z)\|_2\leq \sqrt{N} \|z\|_2$ and $(N-1)/(2S)\geq 2$. We check that the noise assumption \eqref{eq:noisecond} implies that the assumptions in \cite[Theorem 4]{li2020super} hold. Consequently, the output of ESPRIT is a measure $\tilde\mu=\sum_{j=1}^S \tilde a_j \delta_{\tilde t_j}$ such that
	\begin{equation}
		\label{eq:terror}
		\min_\pi \max_j \big|t_j-\tilde t_{\pi(j)}\big|_\T
		\leq \frac{20 S^2}{B} \frac{2^{5/2}\sqrt N}{(N-1)^{3/2}}\, \|H_N(z)\|_2
		\leq \frac{320 S^2}{B M} \, \|z\|_2. 
		\end{equation}
	The optimal permutation is unique. After re-indexing $\tilde\mu$ if necessary, we can assume that $\pi$ is the identity map.
	
	As a consequence of inequality \eqref{eq:terror}, the noise assumption \eqref{eq:noisecond}, and that $\Delta(\mu)\geq \Delta\geq 4/(M-1)$, we have $\Delta(\tilde \mu)\geq \Delta/2\geq 2/(M-1)$. Thus we invoke \cite[Theorem 1.1]{moitra2015matrixpencil} to obtain the lower bound $\sigma_S^2(\tilde \Phi_M) \geq (M-1)/2$. We first see that 
	\[
	\|a-\tilde a\|_2
	=\big\|\tilde\Phi_M^\dagger (\tilde \Phi_M - \Phi_M)a-\tilde\Phi_M^\dagger z\big\|_2
	\leq \sqrt{\frac{2}{M-1}} \( \|a\|_2 \, \big\|\tilde \Phi_M-\Phi_M\big\|_F +\|z\|_2\).
	\]
	A mean value theorem argument shows that 
	\begin{align*}
	\big\|\tilde \Phi_M-\Phi_M\big\|_F
	&= \(\sum_{j=1}^S \sum_{k=0}^{M-1} \big|e^{2\pi ikt_j}-e^{2\pi ik\tilde t_j} \big|^2\)^{1/2} \\
	&\leq \(\sum_{j=1}^S \sum_{k=0}^{M-1} (2\pi k)^{2}|t_j-\tilde t_j|_\T^2\)^{1/2} \leq 640\pi S^2 B^{-1}S^{1/2} \(\frac{M-1}{3}\)^{1/2} \|z\|_2 .  
	\end{align*}
	Combining the above observations, we have that
	\[
	\|a-\tilde a\|_2
	\leq 1280\pi \sqrt{\frac{2}{3}} B^{-1}S^{5/2} \|a\|_2 \|z\|_2. 
	\]
\end{proof}

\begin{remark} 
	Other subspace methods such the matrix pencil method \cite{moitra2015matrixpencil}, MUSIC \cite{liao2016music,li2021stable,li2019conditioning} and another form of ESPRIT \cite{fannjiang2016compressive} enjoy similar robustness properties, but with possibly different implicit constants. These alternative algorithms can also be combined with our beta-quantizer in the same manner, but this paper concentrates on ESPRIT. 
\end{remark} 

This robustness result immediately implies an upper bound on the reconstruction error for when $\calF_M\mu$ is quantized to the alphabet $\calA_K$ by MSQ. Indeed, if $z$ is the vector of quantization errors, then $\|z\|_2\leq \sqrt{M}\|z\|_\infty$. When $K$ is sufficiently large depending on $M$ and $S$, this theorem shows that the amplitude error is at most $O(\lambda^{1/2} K^{-1})$, where the implicit constant is independent of $\lambda$ and $K$.

\subsection{Noise-shaping quantization for ESPRIT}

We show how to combine the our proposed quantization method with ESPRIT to derive an encoding-decoding scheme that offers a significantly better reconstruction accuracy than MSQ. We first quantify how $\calE_{\infty,2}$ is affected by multiplication. 

\begin{lemma}
	\label{lem:Elip3}
	Fix any Lipschitz $\psi$ for which there exists $C\geq 1$ such that $C^{-1}\leq |\psi(t)|\leq C$ for all $t\in\T$. For any atomic $\rho,\nu\in\calM(\T)$ with the same number of atoms,
	\begin{align*}
		\big(1+C |\psi|_{\Lip} + C^2\big)^{-1} \calE_{\infty,2}(\psi\rho,\psi\nu)
		\leq 
		\calE_{\infty,2}(\rho,\nu)
		\leq  \big(1+C |\psi|_{\Lip} + C^2\big) \calE_{\infty,2}(\psi\rho,\psi\nu).
	\end{align*}
\end{lemma}

\begin{proof}
	Let $\rho=\sum_{j=1}^S u_j\delta_{r_j}$ and $\nu=\sum_{j=1}^S v_j\delta_{s_j}$. Note that $\rho$ and $\psi\rho$ have the same support and are both $S$-atomic. Likewise for $\nu$ and $\psi\nu$. Thus, the optimal permutation(s) and the support error in both $\calE_\infty(\rho,\nu)$ and $\calE_\infty(\psi\rho,\psi\nu)$ are identical. After re-indexing $\nu$, we can assume the optimal permutation is the identity map. We also note that
	\begin{equation}
	\label{eq:u}
	\frac{1}{C}\big(\sum_{j=1}^S |\psi(r_j) u_j|^2\big)^{1/2}
	\leq \|u\|_2
	\leq C\big(\sum_{j=1}^S |\psi(r_j) u_j|^2\big)^{1/2}. 
	\end{equation}
	
	We first concentrate on the second claimed inequality of this lemma. We have, 
	\begin{align*}
	\(\sum_{j=1}^S |u_j-v_j|^2\)^{1/2}
	&\leq \(\sum_{j=1}^S \Big|u_j- \frac{\psi(r_j)}{\psi(s_j)} u_j \Big|^2\)^{1/2} + \(\sum_{j=1}^S \Big|\frac{\psi(r_j)}{\psi(s_j)} u_j- v_j \Big|^2\)^{1/2}  \\
	&\leq C |\psi|_{\Lip} \(\max_{j} |r_j-s_j|_\T \) \|u\|_2 + C \(\sum_{j=1}^S |\psi(r_j) u_j- \psi(s_j) v_j |^2\)^{1/2}. 
	\end{align*}
	This together with \eqref{eq:u} implies 
	\begin{align*}
	\calE_{\infty,2}(\rho,\nu)
	&\leq \big( 1+ C |\psi|_{\Lip} \big) \(\max_{j} |r_j-s_j|_\T \) + \frac{C}{\|u\|_2} \(\sum_{j=1}^S |\psi(r_j) u_j- \psi(s_j) v_j |^2\)^{1/2} \\
	&\leq \big( 1+ C |\psi|_{\Lip} + C^2 \big) \calE_{\infty,2}(\psi\rho,\psi\nu).
	\end{align*}
	
	For the first claimed inequality, we have
	\begin{align*}
	&\(\sum_{j=1}^S |\psi(r_j)u_j-\psi(s_j)v_j|^2\)^{1/2} \\
	&\quad\leq \(\sum_{j=1}^S |\psi(r_j)u_j-\psi(s_j)u_j|^2\)^{1/2} + \(\sum_{j=1}^S |\psi(s_j)u_j-\psi(s_j)v_j|^2\)^{1/2} \\
	&\quad\leq |\psi|_{\Lip} \(\max_{j} |r_j-s_j|_\T \) \|u\|_2 + C \(\sum_{j=1}^S |u_j-v_j|^2\)^{1/2}. 
	\end{align*}
	This together with \eqref{eq:u} implies
	\begin{align*}
	\calE_{\infty,2}(\psi\rho,\psi\nu)
	&\leq (1+C|\psi|_\Lip)\(\max_{j} |r_j-s_j|_\T \) + \frac{C^2}{\|u\|_2} \(\sum_{j=1}^S |u_j-v_j|^2\)^{1/2}  \\
	&\leq \big( 1+ C |\psi|_{\Lip} + C^2 \big) \calE_{\infty,2}(\rho,\nu).
	\end{align*} 
	
\end{proof}

With this lemma at hand, the following theorem provides an estimate for the reconstruction accuracy of our encoding-decoding scheme when ESPRIT is used as the super-resolution algorithm in the decoding step. 

\begin{theorem}
	\label{thm:quanesprit} 
	There exists a universal constant $C>0$ such that the following hold. Given $S\geq 1$, even $m\geq 8S$, $K\geq 2$, $\lambda\geq 1$, and real numbers $0<B\leq A$, assume that
	\begin{equation}
	\label{eq:cond1}
	(\lambda+1)K^{-\lambda} 
	\leq \frac{B}{3200 {\rm e} A S^2 m^{3/2}}.
	\end{equation}
	Let $M=m\lambda$ and $(Q,D)$ be the encoder-decoder pair summarized in Section \ref{sec:quan} that uses $ES_{m,n,S}$ with parameter $n:=m/2+1$ as the super-resolution algorithm in the decoding step. 
	
	For any $\Delta \geq 4/(m-1)$ and $\mu\in\calM(\T,\Delta,A,B,S)$, given quantized Fourier measurements $q:=Q(\calF_M\mu)$, the output $\tilde\mu:=D(q)$ satisfies
	\begin{equation*}
	\calE_{\infty,2}(\mu,\tilde\mu)
	\leq C A B^{-1} S^{5/2} M^{3/2} \lambda^{1/2} K^{-\lambda}.
	\end{equation*}
\end{theorem}

\begin{proof}
	
	We start by recalling some essential properties our encoding-decoding scheme. Let $\mu=\sum_{j=1}^S a_j\delta_{t_j}\in \calM(\T,\Delta,A,B,S)$. We have $\|\calF_M\mu\|_\infty\leq \|\mu\|_{\TV}\leq A$ and $V\calF_M\mu=\calF_m(w_m \mu)$. Letting $q:=Q(\calF_M \mu)\in\C^M$, we define the condensed quantization error $z:=Vy-Vq$ and recall that $\|z\|_2\leq \epsilon_V$.

	We first check that the conditions of Theorem \ref{thm:LL} hold for $w_m\mu$ and noisy Fourier coefficients $Vq$. It follows from inequality \eqref{eq:weights} that 
	\[
	\|w_m\mu\|_{\TV}
	\leq \|w\|_\infty \|a\|_1
	\leq A C_\beta,
	\quad\text{and}\quad 
	\min_j |w_m(t_j)a_j|\geq C^{-1}_\beta \min_j |a_j|\geq BC_\beta^{-1}. 	
	\]
	Since $w_m\mu$ and $\mu$ have the same support, we have shown that 
	\[
	w_m\mu\in\calM(\T,\Delta,AC_\beta,BC_\beta^{-1},S).
	\]
	The condition \eqref{eq:cond1} and choice of parameters \eqref{eq:parameters} imply
	\begin{equation*}
	\|z\|_2
	\leq \epsilon_V
	:={\rm e}A \sqrt{2m} (\lambda+1)K^{-\lambda}
	\leq \frac{B\sqrt{2}}{3200S^2m}
	\leq \frac{\Delta}{1280 C_\beta S^2}. 
	\end{equation*}
	This shows that the assumptions of Theorem \ref{thm:LL} are satisfied, so letting $\nu:=ES_{m,n,S}(Vq)$, we have the representation $\nu=\sum_{j=1}^S \tilde b_j \delta_{\tilde t_j}$ and 
	\begin{equation}
	\label{eq:subspace1}
	\calE_\infty(w_m\mu,\nu)
	\leq 320 C_\beta B^{-1} S^2 m^{-1} \epsilon_V + 3400  C_\beta B^{-1} S^{5/2} \epsilon_V.
	\end{equation}
	
	By construction, the output $\tilde\mu:=D(q)$ is defined to satisfy $\tilde\mu= w_m^{-1}\nu$. Combining Lemma \ref{lem:Elip3} applied to $\psi=w_m$, inequality \eqref{eq:subspace1}, and choice of parameters \eqref{eq:parameters3}, we see that 
	\begin{align*}
	\calE_\infty(\mu,\tilde \mu)
	&\leq \big(1+C_\beta |w_m|_\Lip + C_\beta^2 \big) \, \calE_\infty(w_m\mu,\nu) \\
	&\leq \big(1+C_\beta mL_{\beta,\lambda} + C_\beta^2 \big) \big(320 C_\beta B^{-1} S^2 m^{-1} + 3400 C_\beta B^{-1} S^{5/2}  \big) \epsilon_V \\
	&\leq C B^{-1} S^{5/2} M \epsilon_V \\
	&\leq C A B^{-1} S^{5/2} M^{3/2} \lambda^{1/2} K^{-\lambda}. 
	\end{align*}
\end{proof}

\begin{remark}
	Assumption \eqref{eq:cond1} is due to ESPRIT requiring a bound on the noise energy, so any improvements to Theorem \ref{thm:LL} would readily result in weaker assumptions and improved error estimates in Theorem \ref{thm:quanesprit}. This assumption can be interpreted in two ways depending on the circumstances: either the oversampling ratio $\lambda$ is fixed so we need sufficiently many quantization levels $K$, or the number of quantization levels $K$ is not adjustable and we require the over-sampling factor $\lambda$ to be sufficiently large. 
\end{remark}

\section{A general principle}
\label{sec:principle}

Our proposed noise-shaping quantizer was combined with either a convex algorithm or subspace method by following a general template that can be transferred to other robust super-resolution algorithms. Here we explain the main requirements in a general form. 




The first requirement deals with properties of a perhaps non-linear map $\Psi_m\colon \C^m\to\calM(\T)$ and appropriate class of measures $\calM$. Assume that $\Psi_m$ is a left inverse of $\calF_m$ on $w_m\calM$:
\begin{equation}
	\label{eq:principle0}
	\text{for all } \mu\in \calM, \quad
	(\Psi_m\circ \calF_m)(w_m \mu)
	=w_m \mu.
\end{equation}
In particular, we showed that this property holds for when $\Psi_m$ is either a convex or subspace method and for an appropriate $\calM\subset \calM(\T,\Delta)$.

The second requirement concerns properties of a distortion function $\calE$. Assume there is a $C>0$ and $\alpha>0$ such that for any $\mu\in \calM$ and $\xi\in\C^m$, we have 
\begin{equation}
	\label{eq:principle1}
	\calE(w_m\mu,\Psi_m(\xi))
	\leq C \|\calF_m(w_m\mu) - \xi\|_2^\alpha
	= C\|V\calF_M\mu -\xi\|_2^\alpha. 
\end{equation}
This inequality tells us if we view $\xi$ as a perturbation of $\calF_m(w_m\mu)$, then $\Psi_m$ is robust to small perturbations with respect to the distortion function. For convex and subspace methods, we implicitly proved that this inequality holds for $\calE_\Lip$ with $\alpha=1/2$ and $\calE_{\infty,2}$ with $\alpha=1$, respectively. 

The third requirement deals with the robustness of $\calE$ to the re-weighting map $\mu\mapsto w_m\mu$ on $\calM$. Assume there exists a $C>0$ that for any $\mu,\nu\in\calM$, we have
\begin{equation}
	\label{eq:principle2}
	\calE(\mu,\nu)
	\leq C \calE(w_m\mu,w_m\nu). 
\end{equation}
This inequality tells us that the distortion between $\mu$ and $\nu$ can be controlled in terms of the their re-weighted measures. We proved that such an inequality holds for both $\calE_\Lip$ and $\calE_{\infty,2}$ in Lemmas \ref{lem:Elip2} and \ref{lem:Elip3} respectively.

Under these three requirements, we define the map $\Psi_V$ on $\calM$ by $\mu\mapsto w_m^{-1} \Psi_m(\mu)$. We readily verify that $\Psi_V$ is a decoder for $\Phi_V:=V\circ \calF_M$ on $\calM$ because for each $\mu\in\calM$, using \eqref{eq:principle0}, we see that
\[
\Psi_V(\Phi_V \mu)
=w_m^{-1}(\Psi_m(V\calF_M\mu))
=w_m^{-1}(\Psi_m(\calF_m (w_m\mu)))
=w_m^{-1}(w_m \mu)
=\mu.
\]

Let us see what happens if we use this decoder $\Psi_V$ on our beta quantized Fourier coefficients for any $\mu\in\calM$. Recalling our notation that $y=\calF_M\mu$ and $q=Q(y)$, defining $\tilde\mu:=\Psi_V(Vq)$, and using inequalities \eqref{eq:principle1} and \eqref{eq:principle2}, we have
\begin{align*}
\calE(\mu,\tilde\mu)
&\leq C \calE(w_m\mu,\Psi_m(Vq))
\leq C \|Vy-Vq\|_2^\alpha
\leq C \epsilon_V^\alpha.
\end{align*}
Hence, the distortion is controlled in terms of $\epsilon_V$. It is important to mention that the implicit constant $C$ is independent of the quantizer, so in particular, it is independent of both $\epsilon_V$ and $K$.

\section{Numerical results}
\label{sec:num}

\subsection{Methodology}

Theorems \ref{thm:quanconvex} and \ref{thm:quanesprit} are purely deterministic, and thus, they hold for the worst case measure(s). To evaluate their tightness, we search through a vast collection of measures, apply the proposed quantization scheme using TV-min and ESPRIT, compute the error for each one, and take the maximum error over all the simulations. Here, we describe how to select the measures according to a semi-random process. 

Although an explicit formula for the worst case $\mu$ is unknown, it seems intuitive that for fixed separation $\Delta$ and noise energy, the greater the number of atoms $S$, the harder it would be to resolve the measure. We fix $\Delta =0.15$, $m=\lceil 4/\Delta \rceil$, $t_1=0$, and $t_2=\Delta$. For each $j\geq 3$, we select $t_{j+1}=t_j+\Delta+|\eta_j|$, where $\eta_j$ is a Gaussian random variable with mean 0 and standard deviation $\Delta$, and this process terminates at step $S$, when $t_{S+1}>1-\Delta$. The stochastic term $\eta_j$ is included to avoid selecting a set with special number-theoretic properties. This provides us with a generic set $T=\{t_j\}_{j=1}^S$ containing with $S$ points (where $S$ may vary depending on the random draw of the $\eta_j$'s) with minimum separation at least $\Delta$. 

After the support set $T$ is chosen according to this process, we select the amplitudes $a\in\C^S$ such that each $a_j$ is drawn from the uniform distribution on the complex circle of radius $1/S$. Hence, $B=1/S$ and $A=1$. We vary the phases of the entries of $a$ because TV-min performs differently on complex versus positive measures (e.g. \cite{duval2015exact,benedetto2020super}). On the other hand, ESPRIT does not depend on the phases because it only depends on the subspace which are invariant under phase shifts (e.g. \cite{li2020super}). Based on this semi-random process, we obtain a $S$ and complex measure 
$$
\mu=\sum_{j=1}^S a_j \delta_{t_j} \in \calM(\T,\Delta,1,1/S,S). 
$$

The free parameters in our experiments are the over-sampling ratio $\lambda$ and the number of levels $K$ for the real and imaginary components. We construct $100$ measures chosen according to the aforementioned semi-random method. For each one, we compute its exact $M=\lambda m$ Fourier coefficients. We encode the measurements using both MSQ and the proposed beta-encoder, both with the same alphabet, and then decode the quantized measurements, as outlined in Section \ref{sec:beta}. To measure their reconstruction errors, we compute the amplitude error. Finally, for each $\lambda$ and $K$, we find the maximum error over these $100$ trials. 

\subsection{Illustration of the theory}

\begin{figure}[!]
	
	\caption{Reconstruction error for TV-min and ESPRIT using MSQ and $\beta$-quantization. The results are shown as functions of $K$ and $\lambda$. The dashed black lines are visual guides.}
	
	\begin{subfigure}{\textwidth}
		\centering
		\includegraphics[width=0.5\textwidth]{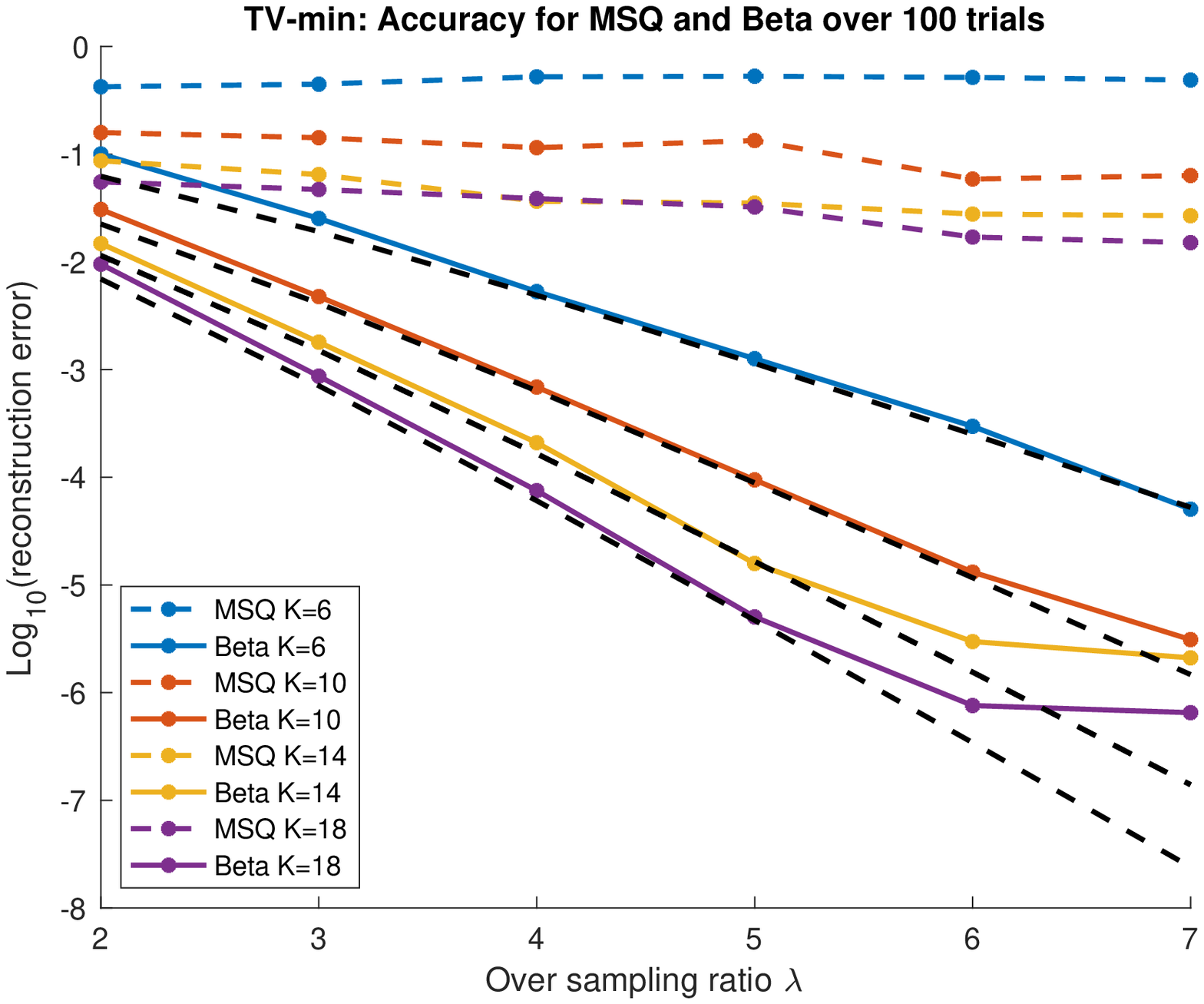}
		\hspace{-2em}
		\includegraphics[width=0.5\textwidth]{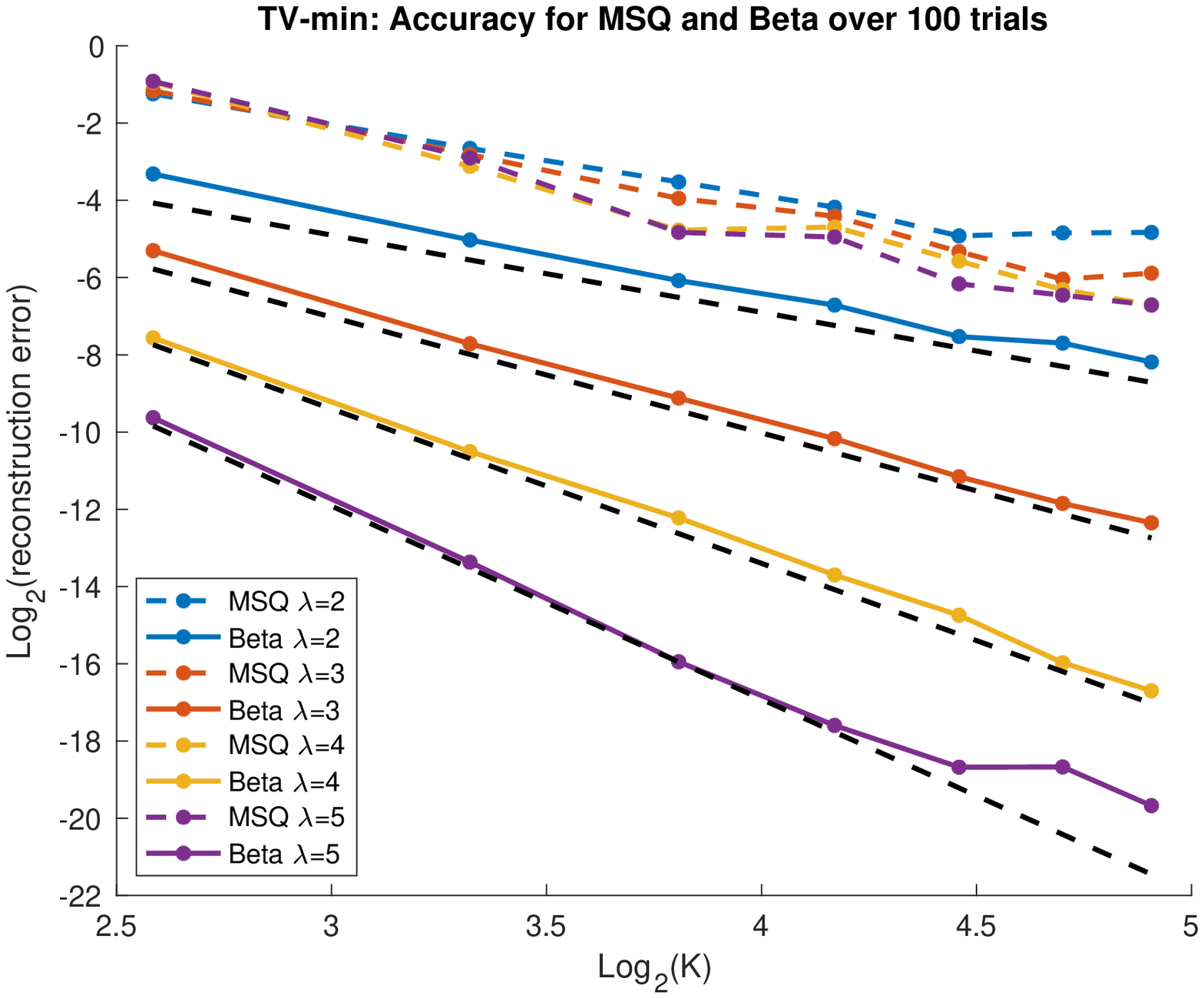}
		\caption{Plot of the amplitude error for TV-min, using MSQ and $\beta$-quantization, for various choices of $K$ and $\lambda$. The dashed black lines are the graphs of $0.75\cdot  \lambda^{3/2} K^{-\lambda}$.}
		\label{fig:exper1}
	\end{subfigure}
	
	\begin{subfigure}{\textwidth}
		\centering
		\includegraphics[width=0.5\textwidth]{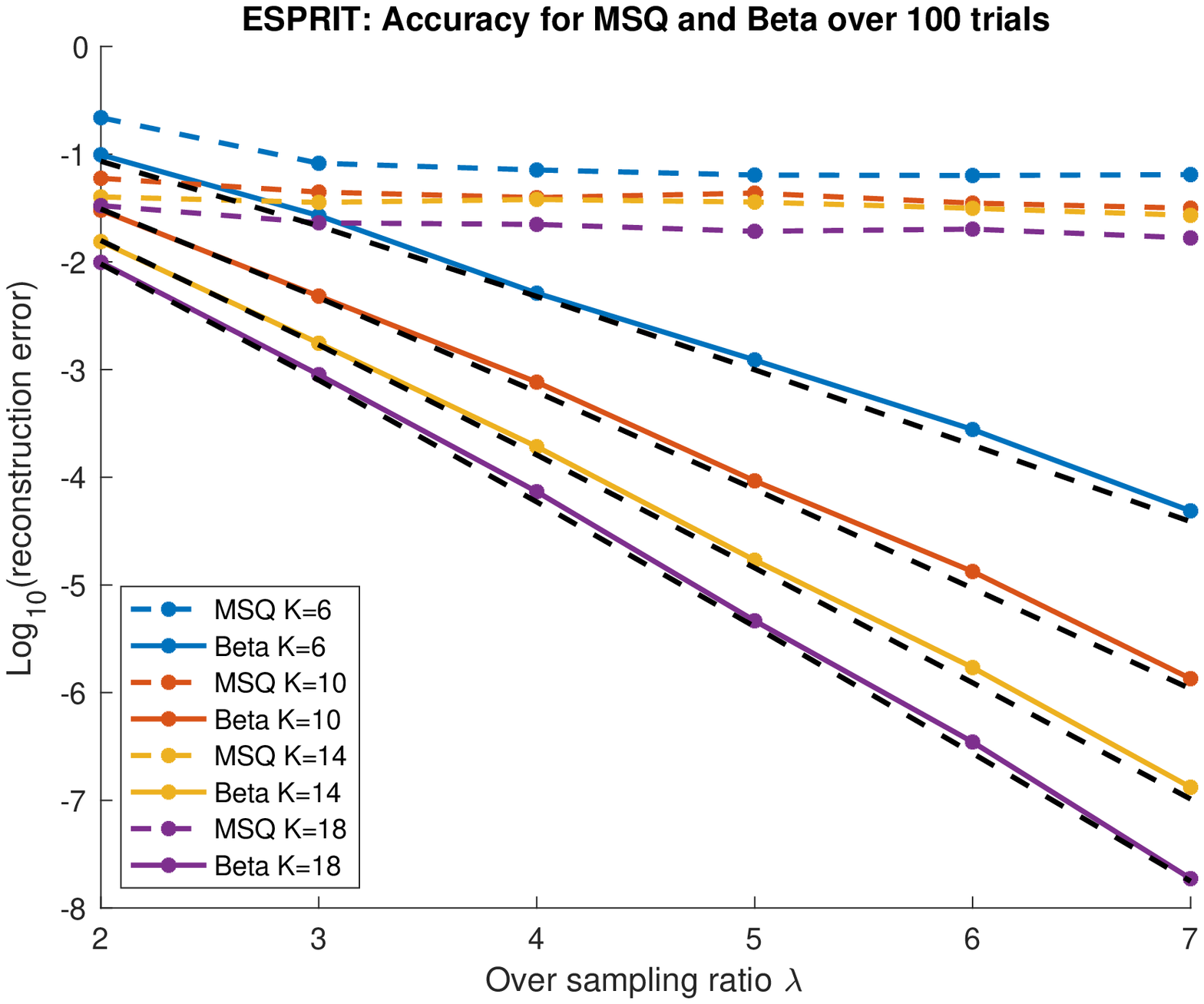}
		\hspace{-2em}
		\includegraphics[width=0.5\textwidth]{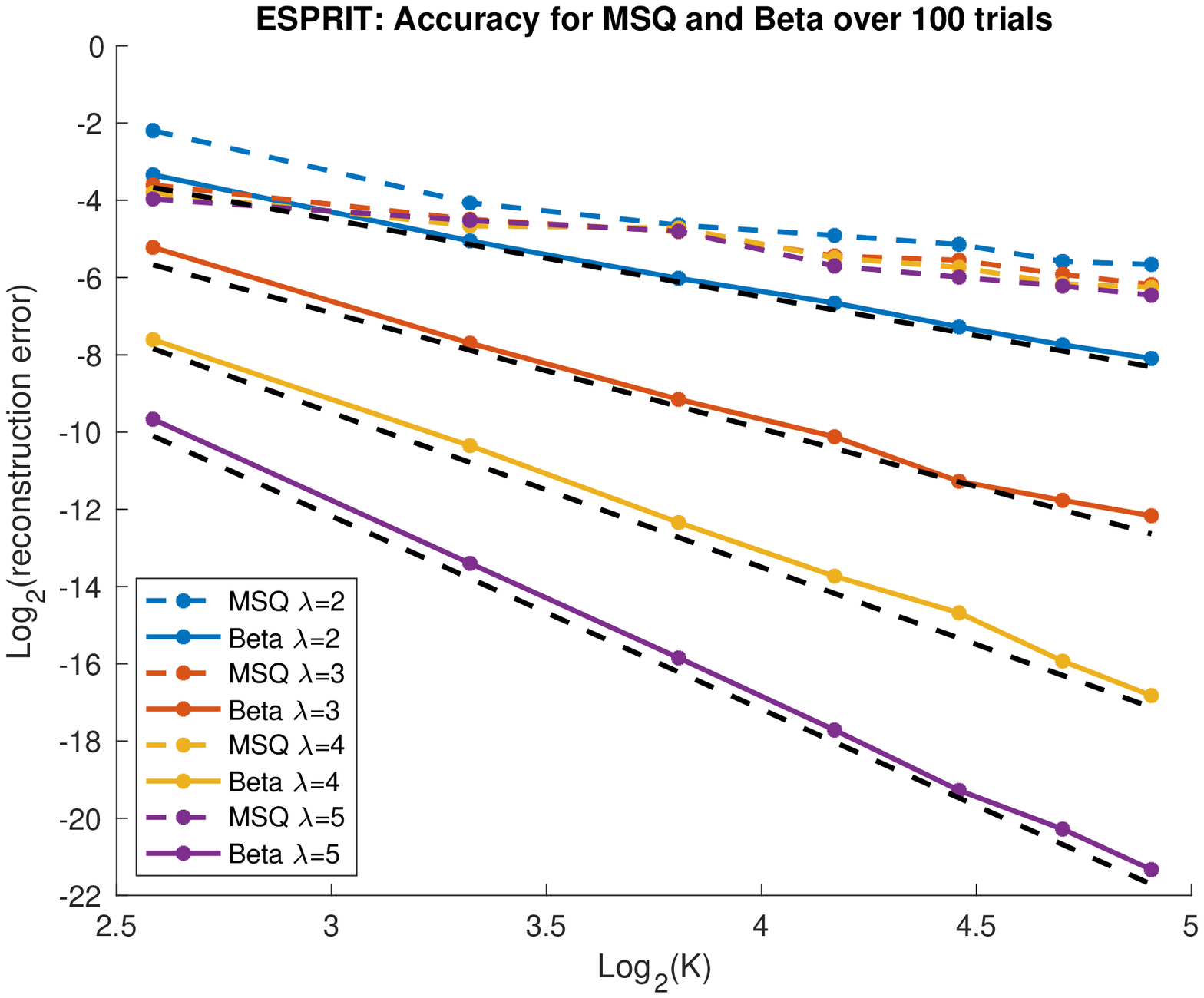}
		\caption{Plot of the amplitude error for ESPRIT, using MSQ and $\beta$-quantization, for various choices of $K$ and $\lambda$. The dashed black lines are the graphs of $1.6\cdot  \lambda K^{-\lambda}$.}
		\label{fig:exper2}
	\end{subfigure}

\end{figure}

We first concentrate on TV-min. Since is not straightforward to numerically compute $\calE_\Lip$, we instead measure the error in terms of $\calE^m_1$, which is the dominant error terms in Theorem \ref{thm:quanconvex}, and also see Theorem \ref{thm:quanconvex2}. When either $K$ or $\lambda$ is sufficiently large, for some $C>0$ independent of $K$ and $\lambda$, satisfies
\[
\log(\text{distortion})\leq \log C+\frac{3}{2}\log \lambda - \frac{\lambda}{2}\log K. 
\]
To obtain the above expression, note that $M=\lambda m=O(\lambda)$ since $m$ is fixed. If we view $K$ as fixed, then the log accuracy is essentially proportional to $\lambda$ with slope $-\log K$. On the other hand, if we view $\lambda$ as a fixed constant, then the log accuracy is proportional to $\log K$ with slope $-\lambda/2$ and intercept $\log C+3(\log\lambda)/2$.

The numerical simulations shown in Figure \ref{fig:exper1} indicate that the amplitude error for TV-min, combined with our encoding-decoding scheme, is not consistent with the theory when $K^{-\lambda}$ is greater than about $10^{-7}$. In this regime, the numerical results indicate that the actual reconstruction accuracy for TV-min using our encoding-decoding scheme should be $O(\lambda^{3/2}K^{-\lambda})$. On the other hand, when $K^{-\lambda}$ is smaller than about $10^{-7}$, the reconstruction error seems to saturate. One explanation is TV-min undergoes a phase transition, where perhaps asymptotically, the correct dependence is $O(\lambda^{3/2}K^{-\lambda/2})$. Another possibility is that when $K^{-\lambda}$ is small, the feasible set in TV-min has radius $\epsilon_V$, so optimizing over this region might lead to numerical errors. 

We proceed to discuss the numerical simulations for ESPRIT. While Theorem \ref{thm:quanesprit} is stated in terms of the $\calE_{\infty,2}$ error, to keep the subsequent experiments consistent with that of the TV-min ones, we will numerically calculate the $\ell^2$ amplitude error, which is the dominant term in Theorem \ref{thm:quanesprit}. When either $K$ or $\lambda$ is sufficiently large, for some $C>0$ independent of $K$ and $\lambda$, should satisfy the relationship
\[
\log(\text{distortion})
\leq \log C+2\log \lambda - \lambda\log K.
\]
To obtain the above expression, note that $M=\lambda m=O(\lambda)$ since $m$ is fixed. Again, we can interpret this inequality from two perspectives. If we view $K$ as a fixed constant, then the log accuracy is essentially proportional to $\lambda$ with slope $\log K$. On the other hand, if we view $\lambda$ as a fixed constant, then the log accuracy is proportional to $\log K$ with slope $-\lambda$ and intercept $\log C+2\log \lambda$. The simulations shown in Figure \ref{fig:exper2} suggest that the true distortion for our encoding-decoding method combined with ESPRIT should behave as $O(\lambda K^{-\lambda})$.

\subsection{Pushing the limits of quantization/extreme cases}

There is significant interest in coarse quantization schemes. Since Fourier measurements are inherently complex-valued, it is natural to quantize the real and imaginary parts separately, as we have done in this paper. Thus, the smallest alphabet we consider contains $K^2= 4$ elements, which amounts to one-bit quantization for the real and imaginary parts. It is possible to reduce the alphabet from four to three levels by employing a triangular lattice (e.g. \cite{chou2017distributed}), but we do not consider this situation here. 

When $K$ is small, the only way in which we can achieve reasonable reconstruction error is to have reasonably large over-sampling ratio $\lambda$, in order to exploit any redundancy in the measurements. Figure \ref{fig:coarse1} shows the reconstruction accuracy for our encoding-decoding schemes that uses TV-min and ESPRIT, for small values of $K$ as a function of $\lambda$. The numerical simulations show that our method can handle coarsely quantized measurements provided that the over-sampling is sufficiently large. In the important $K=2$ case, the results show that the reconstruction error decays as $O(2^{-\lambda})$. 

Figure \ref{fig:coarse2} shows the reconstruction error when $\lambda$ is relatively small. As expected, without highly redundant information, the reconstruction accuracy is poor. It appears that TV-min provides better results in the small $\lambda$ regime, which is consistent with our theory -- recall that Theorem \ref{thm:quanconvex} does not require an apriori upper bound on $K^{-\lambda}$, while Theorem \ref{thm:quanesprit} does. The reason is that for ESPRIT, the amplitudes are reconstructed using the pseudo-inverse $\Phi_M(\tilde T)^\dagger$, and if $\tilde T$ is not identified correctly, then the pseudo-inverse provides unpredictable results.

\begin{figure}[!]
	\caption{Plot of the reconstruction error for TV-min and ESPRIT for $K=2,3,4$ as a function of $\lambda$. The dashed black lines are visual guides and are the graphs of $1.6\cdot \lambda K^{-\lambda}$.}
	\begin{subfigure}{0.5\textwidth}
		\centering
		\includegraphics[width=\textwidth]{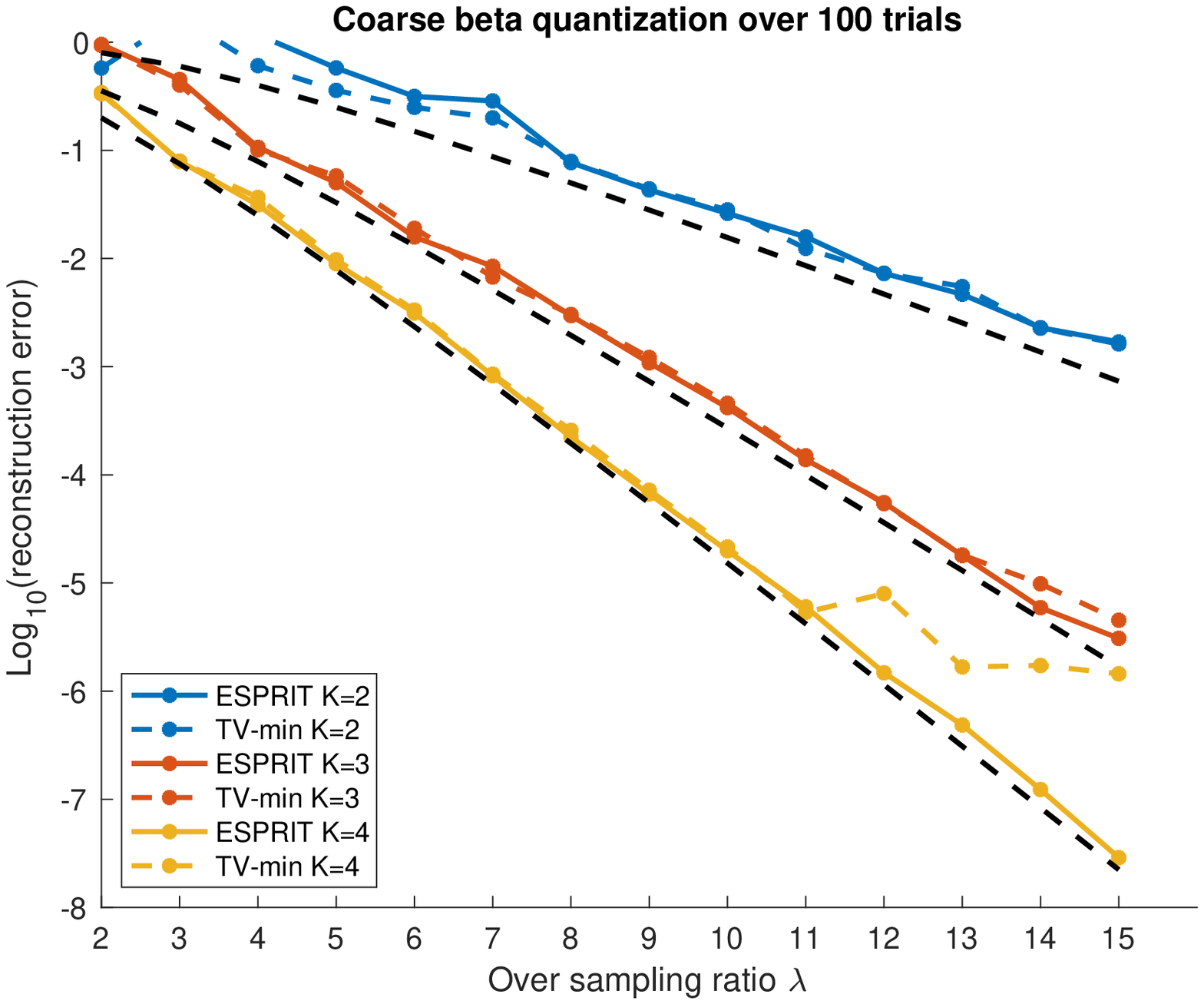}
		\caption{Large $\lambda$ regime}
		\label{fig:coarse1}
	\end{subfigure}
	\begin{subfigure}{0.5\textwidth}
		\centering
		\includegraphics[width=\textwidth]{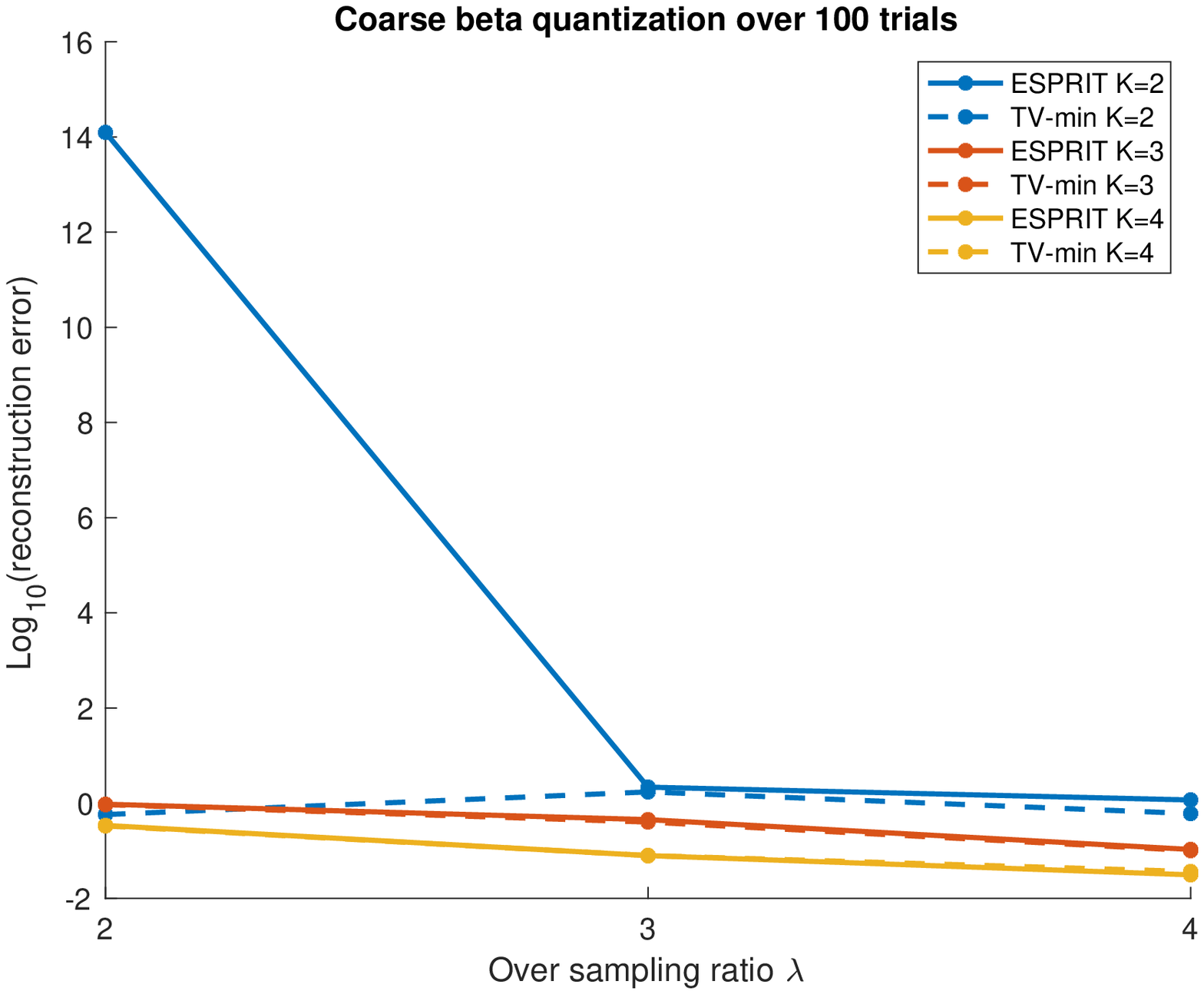}
		\caption{Small $\lambda$ regime}
		\label{fig:coarse2}
	\end{subfigure}
	
\end{figure}

\section{Fundamental limits of MSQ}
\label{sec:lower}

We present lower bounds for the reconstruction error from memoryless scalar quantized finite frame coefficients that are universal over the alphabet and sampling scheme. We discuss its consequences to super-resolution quantization after the results are presented. 

Consider a spanning set of vectors $\{\psi_j\}_{j=1}^M\subset\C^S$ (e.g. a frame) and the $M\times S$ complex matrix $\Psi$ whose $j$-th row is the vector $\psi_j^*$. For any $z\in\C^S$, the frame coefficients of $z$ are $\Psi z=\{\<\psi_j,z\>\}_{j=1}^M$, and it is possible to reconstruct $z$ from $\Psi z$. In this section, we use $\<\cdot,\cdot\>$ to denote the complex inner product on $\C^S$ with the convention that $\langle x, y\rangle = x^*y$. 

Let $\calA\subset\C$ be an alphabet of cardinality $L$. Here, $\calA$ is arbitrary and is not necessarily a Cartesian alphabet like the one defined in \eqref{eq:alphabetdelta}. We let $Q\colon \C^M\to \calA^M$ denote MSQ, which rounds each entry of $z\in\C^M$ to its nearest element in $\calA$ with respect to the standard metric $|\cdot|$ on $\C$. If there is a tie, pick one of the closest elements arbitrarily. 

We would like to answer the question: what is the best possible reconstruction error given only the memoryless scalar quantized finite frame coefficients $Q(\Psi z)$? To make this question precise, let $B_S$ denote the unit $\ell^2$ ball of $\C^S$ and we define the quantity,
\[
\calE(\Psi,\calA)
:=\inf_{D\colon\calA^M\to\C^S} \sup_{z\in B_S} \|z-D(Q(\Psi z))\|_2. 
\]
This quantity describes the error incurred by the best possible decoder, where the error is measured uniformly over $B_S$. The following theorem provides a lower for $\calE(\Psi,\calA)$ that only depends on $L,M,S$. It asserts that no matter which quantization alphabet and frame are used, MSQ suffers a fundamental lower bound.

\begin{theorem}
	\label{thm:lower1}
	For any $M\times S$ frame $\Psi$ and complex alphabet $\calA\subset\C$ of cardinality $L$,
	\[
	\calE(\Psi,\calA)
	\geq \min\(\frac{1}{4},\frac{S}{3{\rm e}LM}\).
	\]
\end{theorem}

\begin{proof}
	Our strategy for analyzing this quantity is to recast this problem as a discrete geometric one. For each $q\in\calA^M$, we define the quantization cell 
	\[
	C(q):=
	\{z\in B_S \colon Q(\Psi z)=q\}. 
	\]
	The significance of $C(q)$ is that this contains all points in $B$ that are mapped to $q$, so any decoder produces the same output for any $z,w\in C(q)$. Then
	\[
	\calE(\Psi,\calA)
	\geq \frac{1}{2} \sup_{q\in\calA^M} \sup_{w,z\in C(q)} \|w-z\|_2.
	\]
	Let $\rho>0$ be the smallest real number such that each quantization cell $C(q)$ can be covered by a ball of radius $\rho$. Then $\calE(\Psi,\calA)\geq \rho$ and the remainder of this proof focuses on lower bounding $\rho$ by estimating the total number of quantization cells, which we denote by $N_C$, and a volumetric argument.  
	
	The alphabet $\calA\subset\C$ induces a {\it Voronoi partition} of $\C$. That is, there exists a collection of sets $\{V(a)\colon a\in\calA\}\subset\C$ called {\it Voronoi cells} that are pairwise disjoint except on sets of Lebesgue measure zero, and $w\in V(a)$ if and only if $|w-a|=\min_{b\in\calA} |w-b|$. The boundary of each $V(a)$ is a polygon with $E(a)$ edges, and the following is a well known upper bound for the total number of edges induced by the partition,
	\begin{equation}
	\label{eq:edges}
	\sum_{a\in\calA} E(a)\leq
	\begin{cases}
	\ 3L-6 &\text{if } L\geq 3,\\
	\ 1 &\text{if } L=2. 
	\end{cases}
	\end{equation}

	Since each Voronoi cell $V(a)$ is a polygon with $E(a)$ edges, membership in $V(a)$ can characterized by $E(a)$ affine inequalities. There exist real numbers $\alpha_k(a),\beta_k(a),\gamma_k(a)$ such that $w\in V(a)$ if and only if
	\[
	\Re(w)\alpha_k(a)+\Im(w)\beta_k(a) \leq \gamma_k(a), \quad k=1,\dots,E(a).
	\]
	From here, we see that $z\in C(q)$ if and only if $\<\psi_j,z\>\in V(q_j)$ for $j=1,\dots,M$, which is equivalent to  
	\begin{eqnarray*}	\Re(\<\psi_j,z\>)\alpha_k(q_j)+\Im(\<\psi_j,z\>)\beta_k(q_j) \leq \gamma_k(q_j), \quad k=1,\dots,E(q_j), \quad j=1,\dots,M.
	\end{eqnarray*}
	Let $\psi_j^R,\psi_j^I\in\R^S$ be the real and imaginary parts of $\psi_j\in\C^S$ respectively, and let $x,y\in\R^S$ be the real and imaginary parts of $z\in\C^S$ respectively. The above is equivalent to
	\[
	\Big\< 
	\begin{bmatrix}
	\psi_j^R\alpha_k(q_j) + \psi_j^I\beta_k(q_j) \\
	-\psi_j^I \alpha_k(q_j)+\psi_j^R \beta_k(q_j)
	\end{bmatrix},
	\begin{bmatrix}
	x \\
	y
	\end{bmatrix}
	\Big\>
	\leq \gamma_k(q_j), 
	\quad k=1,\dots,E(q_j), \quad j=1,\dots,M.
	\]
	We interpret this is a collection of hyperplane inequalities in $\R^S\times \R^S$. We have proved that all $N_C$ quantization cells can be determined by at most $N_P$ hyperplanes, where
	\begin{equation*}
	N_P
	\leq M \sum_{a\in\calA} E(a) 
	\leq 3LM,
	\end{equation*}
	and the last inequality follows from \eqref{eq:edges}. 
	
	It is known that the maximum number of cells determined by $N_P$ hyperplanes in $\R^S\times\R^S$ is $\sum_{j=0}^{2S} {{N_P}\choose{j}}$. We consider two separate cases.
	\begin{enumerate}[(a)]
		\item 
		If $3LM<4S$, then we have $N_C\leq 2^{N_P}\leq 2^{4S}=4^{2S}$.
		\item 
		If $3LM\geq 4S$, then 
		\[
		N_C
		\leq \sum_{j=0}^{2S} {{N_P}\choose{j}}
		\leq (2S+1) {{3LM}\choose{2S}}
		\leq (2S+1) \(\frac{3{\rm e} LM}{2S}\)^{2S}.
		\]
		 
	\end{enumerate}
	
	Since there are $N_C$ quantization cells each covered by a ball of radius $\rho$, we have 
	\[
	\vol(B_S)
	\leq N_C \rho^{2S} \vol(B_S). 
	\]
	Combining this inequality, together with our earlier upper bound for $N_C$ and the inequality $(2S+1)^{1/(2S)}\leq 2$, completes the proof.
\end{proof}

To see how the theorem pertains to super-resolution, suppose $\mu$ has $S$ atoms and an oracle provides us the support of $\mu$ for free. Then $\calF_M\mu=\Phi_M(T) a$, where $a\in\C^S$ and $T$ denote the amplitudes and support of $\mu$ respectively, and $\Phi_M$ is the Fourier matrix defined in \eqref{eq:vander}. Notice that we access to $\Phi_M(T)$ due to the oracle giving us $T$ and $\Phi_M(\T)\in\C^{M\times S}$ has linearly independent columns because it is a Vandermonde matrix with unique nodes. The theorem tells us that the best distortion achievable using memoryless scalar quantized Fourier measurements, regardless of the reconstruction algorithm and even with the help of an oracle, is lower bounded by $O(M^{-1} L^{-1})$. 

On the other hand, we showed that for two popular super-resolution algorithms, the reconstruction accuracy from MSQ measurements is upper bounded $O(\sqrt M K^{-1})$ where $K^2$ is the total number of levels. We provide two explanations for the gap between upper and lower bounds. First is that the support estimation part of super-resolution is typically the most difficult part, which was circumvented by the oracle. In this sense, the lower bound has a significant advantage over the upper bound. Second is that Theorem \ref{thm:lower1} applies to general complex alphabets, whereas for the upper bounds, we worked with conceptually simpler and more natural Cartesian alphabets. The following result address the second point.

\begin{theorem}
	\label{thm:MSQ2}
For any $M\times S$ frame $\Psi$ and complex alphabet $\calA\subset\C$ of the form $\calA=\calA_\R+i\calA_\R$, where $\calA_\R\subset\R$ and has cardinality $K$,
\[
\calE(\Psi,\calA)
\geq \min\(\frac{1}{4},\frac{S}{2{\rm e}KM}\).
\]	
\end{theorem}

\begin{proof}
The starting point and overall strategy of this proof is similar to that of Theorem \ref{thm:lower1}. There we proved that if $N_C$ is the number of quantization cells induced by the alphabet $\calA$, then
\[
\calE(\Psi,\calA)
\geq N_C^{-1/(2S)}.
\]
Since $\calA$ is Cartesian by assumption, we see that all Voronoi cells $\{V(a)\colon a\in\calA\}$ can be specified by $K-1$ vertical and $K-1$ horizontal hyperplanes. Thus, if $N_P$ denotes the number of hyperplanes  in $\R^S\times\R^S$ necessary to specify all quantization cells, then 
\[
N_P
=2(K-1)M\leq 2KM.
\]
The maximum number of cells determined by $N_P$ hyperplanes is $\sum_{j=0}^{2S} {{N_P}\choose{j}}$. Repeating a similar argument as in Theorem \ref{thm:lower1}, we have
\[
N_C
\leq 
\begin{cases}
\ \displaystyle (2S+1) \(\frac{{\rm e}KM}{S} \)^{2S} &\text{if } 2KM \geq 4S, \\
\ 4^{2S} &\text{if } 2KM < 4S. 
\end{cases}
\]
Combining these observations and using that $(2S+1)^{1/(2S)}\leq 2$ completes the proof.
\end{proof}

In the context of super-resolution, this theorem demonstrates that even if the support of $\mu$ is provided to us by an oracle, reconstruction from its memoryless scalar quantized Fourier measurements on a Cartesian complex alphabet of cardinality $K^2$ incurs error of at least $O(M^{-1}K^{-1})$ regardless of how the alphabet is spaced and what decoder is used.

\section*{Acknowledgments} 

Weilin Li gratefully acknowledges support from the AMS--Simons Travel Grant. 
 
\bibliography{QuanSRbib}
\bibliographystyle{plain}

\end{document}